\documentclass[12pt,a4paper]{article}
\usepackage[english]{babel}
\usepackage{csquotes}
\usepackage[backend=biber,style=numeric,language=english,sorting=nyt,
minnames=7,maxnames=7]{biblatex}
\usepackage{geometry}

\usepackage{graphicx}
\usepackage{caption}[2022/03/01]
\usepackage{subcaption}
\usepackage{booktabs}
\usepackage{float}

\usepackage{booktabs}
\usepackage{enumerate}
\usepackage{enumitem}

\usepackage{amsmath}
\usepackage{amsthm}
\usepackage{amssymb}
\usepackage{bm}
\usepackage{bbm}
\usepackage{siunitx}

\usepackage{CustomCommands}

\usepackage{hyperref}

{
\title{\bfseries\huge Reconstructing wind fields from gravitational data on gas giants:
An investigation of mathematical methods.}
\author{Tim-Jonas Peter\textsuperscript{1}, Volker Michel\textsuperscript{1}, Franz-Theo Suttmeier\textsuperscript{2}, Johannes Wicht\textsuperscript{3}}
}
\date{}

\makeatletter
\def\th@plain{%
  \thm@notefont{}
  \itshape 
}
\def\th@definition{%
  \thm@notefont{}
  \normalfont 
}
\makeatother

\theoremstyle{plain}

\newtheorem{theoremS}{Theorem}[section]

\newtheorem{lemmaS}[theoremS]{Lemma}

\theoremstyle{definition}

\newtheorem{definitionS}[theoremS]{Definition}

\newtheorem{assumptionS}[theoremS]{Assumption}

\theoremstyle{remark}

\newtheorem{remarkS}[theoremS]{Remark}

\numberwithin{equation}{section}

\begin{document}
\maketitle
\footnotetext[1]{Geomathematics Group, Department of Mathematics, University of Siegen, Emmy-Noether-Campus, Walter-Flex-Str. 3, 57068 Siegen,
Germany, E-Mail addresses: \href{mailto:tim-jonas.peter@mathematik.uni-siegen.de}{tim-jonas.peter@mathematik.uni-siegen.de}, \href{mailto:michel@mathematik.uni-siegen.de}{michel@mathematik.uni-siegen.de}}
\footnotetext[2]{Scientific Computation Group, Department of Mathematics, University of Siegen, Emmy-Noether-Campus, Walter-Flex-Str. 3, 57068 Siegen,
Germany, E-Mail address: \href{mailto:suttmeier@mathematik.uni-siegen.de}{suttmeier@mathematik.uni-siegen.de}}
\footnotetext[3]{Planetary Science Department, Max Planck Institute for Solar System Research,  Justus-von-Liebig-Weg 3, 37077 Göttingen, Germany, E-Mail address: \href{mailto:wicht@mps.mpg.de}{wicht@mps.mpg.de}}

\begin{abstract}
    The atmospheric structure of gas giants, especially those of Jupiter and Saturn, has been an object of scientific studies for a long time. 
    The measurement of the gravitational fields by the Juno mission for Jupiter and the Cassini mission for Saturn offered new possibilities to study the interior structure of these planets.
    Accordingly, the reconstruction of the wind velocities from gravitational data on gas giants has been the subject of many research papers over the years, yet the mathematical foundations
    of this inverse problem and its numerical resolution have not been studied in detail. 
    This article suggests a rigorous mathematical theory for inferring the wind fields of 
    gas giants. 
    In particular, an orthonormal basis is derived which can be associated to 
    models of the gravitational potential and the interior wind velocity field. 
    Moreover, this approach provides the foundations for existing resolution concepts of the inverse problem.
\end{abstract}

\vspace*{2cm}
\newpage

\section{Introduction}
A long standing subject of research in the planetary sciences has been the study of the interior dynamics of gas-giant-atmospheres, like those of Jupiter and Saturn. 
The surfaces of these planets can be observed by telescope --- which has been going on since Galilei observed Jupiter and Saturn in the early 17th century --- and today by spacecraft. 
They display a clearly visible banded structure which is related to fast winds in eastward and westward directions (see Figure \ref*{fig:SurfacWinds} in Section \ref*{sec:InvProb} for some pictures). 
These winds, which circumvent the planets, are called zonal winds. 
Being observed in the cloud layer at a pressure of a few bar, it was an open question how deep these winds may extend into the interior and how strong they are at a specified depth \cite{CylindWindsKaspi,AtmosWindsKaspi,ZonalWindsWicht,DiffRotCompWisdom}.
To attempt to answer these questions, further data related to the interior of Jupiter and Saturn is required.
This data was provided by the Juno mission around Jupiter, which arrived in orbit around the planet in 2016, and the Cassini Grand Finale, which took place in 2017, for Saturn. 
These spacecrafts measured the gravitational field of their respective target planets (see \cite{JupiterGravFieldIess,SaturnGravFieldIess} for the data on Jupiter and Saturn respectively), with the hope being that this data would provide information about mass transports caused by wind fields in the planetary interior.
This would enable conclusions about the structure of the atmosphere and wind fields.
\par
More typically, gravity data are inverted for the density structure inside a planet. 
At Earth, very precise high resolution data are routinely used to identify geological features or to search for ore deposits. 
For further details on the mathematical formulation, see the survey article \cite*{NonUniqueGravMichel} and the references therein. 
For gas planets, the very fast zonal winds yield tiny density variations whose gravity signal can be measured by modern space missions, this interrelation has been studied in several papers \cite{ZonalWindsWicht2,GravSignZonalFlowsHubbard,CylindWindsKaspi,AtmosWindsKaspi,ZonalWindsWicht,DiffRotCompWisdom}
The gravity data suggest that the zonal winds reach about 3000 km deep in Jupiter and about 9000 km deep in Saturn. 
However, the inversion of gravity for flow is complex and relies on a number of assumptions.
The aim of this article is to examine the methods that allow inverting zonal properties from gravity data from a rigorous mathematical viewpoint. 
We will examine the models for the gravitational potential of gas giants in Section \ref{sec:GravPot} and for the wind field in Section \ref{sec:WindInsidePlanet}. 
After that an orthonormal basis adapted for this problem is developed in Section \ref{sec:ONBGravPot} and finally we show how this basis can be used to connect the gravitational data to the wind fields in Section \ref{sec:JnConnect}.
This paper is in large part based on the thesis \cite{MasterthesisTim}.

\section{On Notation}
In $\R^3$, we label the standard basis vectors by $\bm{e}^1,\bm{e}^2$ and $\bm{e}^3$. Some
sets that we will use are 
\begin{align*}
    B &:= B_R(0) = \{\bx\in\R^3\,|\, |\bx|<R\} \\
    \S^2 &:= \{\bx\in\R^3\,|\, |\bx|=1\}\text{ and }\\
    \partial B &= \{\bx\in\R^3\,|\, |\bx|=R\}.
\end{align*}
$\overline{B}$ will refer to the closure of $B$, meaning that $\overline{B} = \{\bx\in\R^3\,|\,|\bx|\leq R\}$.
We will also use polar coordinates, which we can use to express any vector 
$\bx\neq 0$ in the following form
\begin{equation*}
    \bx = \begin{pmatrix}
        r\sqrt{1-t^2}\cos(\varphi) \\
        r\sqrt{1-t^2}\sin(\varphi) \\
        rt
    \end{pmatrix},
\end{equation*}
where $r>0$, $\varphi\in[0,2\pi)$ and $t\in[-1,1]$. Here, the polar distance $t$ can be 
expressed in terms of a colatitude $\vartheta\in[0,\pi]$, by $t=\cos(\vartheta)$ or in terms of a latitude $\theta\in[-\pi/2,\pi/2]$, by $t=\sin(\theta)$. This allows 
us to construct a local orthonormal basis of $\R^3$ via
\begin{equation*}
    \bm{e}^r :=\begin{pmatrix}
        \sqrt{1-t^2}\cos(\varphi) \\
        \sqrt{1-t^2}\sin(\varphi) \\
        t
    \end{pmatrix},\,
    \bm{e}^{\varphi} :=\begin{pmatrix}
        -\sin(\varphi) \\
        \cos(\varphi) \\
        0
    \end{pmatrix}\text{ and }
    \bm{e}^{t} :=\begin{pmatrix}
        -t\cos(\varphi) \\
        -t\sin(\varphi) \\
        \sqrt{1-t^2}
    \end{pmatrix}.
\end{equation*}
We can also decompose any vector $\bx\neq 0$ as $\bx = r\bi$, where $r:=|\bx|$
and $\bi:=\bx/|\bx|\in\S^2$. In polar coordinates we would write 
$\bx=r\bm{e}^r$. \\
This decomposition of any non-zero vector into a radial and a tangential part can also 
be applied to differential operators, yielding the following definitions. For a 
continuously differentiable function $F$, the \textbf{surface gradient operator} $\nabla^{\ast}$
is defined via
\begin{equation*}
    \nabla F(\bx) = \bi\partial_rF(r\bi)+\frac{1}{r}\nabla^{\ast}F(r\bi).
\end{equation*}
In polar coordinates this means that
\begin{equation*}
    \nabla^{\ast} = \sqrt{1-t^2}\bm{e}^t\partial_t+\frac{1}{\sqrt{1-t^2}}\bm{e}^{\varphi}
    \partial_{\varphi}.
\end{equation*}
We will also need the \textbf{surface curl operator}, defined by 
\begin{equation*}
    \Lup^{\ast}F(r\bi) := \bi\times\nabla^{\ast}F(r\bi),
\end{equation*}
or in polar coordinates
\begin{equation*}
    \Lup^{\ast} = -\sqrt{1-t^2}\bm{e}^{\varphi}\partial_t+\frac{1}{\sqrt{1-t^2}}
    \bm{e}^t\partial_{\varphi}.
\end{equation*}
If the function $F$ is twice continuously differentiable, then we can define the 
\textbf{Laplace-Beltrami operator} $\Delta^{\ast}$ by decomposing the \textbf{Laplacian}
$\Delta$ into a radial and a tangential part.
\begin{equation*}
    \Delta F(r\bi) = \partial_r^2F(r\bi)+\frac{2}{r}\partial_rF(r\bi)+
    \frac{1}{r^2}\Delta^{\ast}F(r\bi),
\end{equation*}
then $\Delta^{\ast}$ only acts on the coordinate $\bi\in\S^2$.\\
We will also require some function spaces, namely the space of all real functions which 
are continuously differentiable $k$ times on an open subset $D\subs\R^3$, denoted by 
$\Cont^k(D)$. The spaces of square integrable real functions on the open ball
of radius $R>0$ and on the unit sphere, denoted by $\LSq(B)$ and $\LSq(\S^2)$, 
respectively will be needed as well. Finally $\LSq_w(0,R)$ will refer to the real, 
weighted $\LSq$-space with weight function $w(r)>0$ on $(0,R)$ and inner product
\begin{equation*}
    \langle F,G\rangle_{\LSq_w(0,R)} := \int_0^RF(r)G(r)w(r)\di r. 
\end{equation*}
For this weighted $\LSq$-space the norm will be given by 
\begin{equation*}
    \|F\|_{\LSq_w(0,R)} := \int_{0}^{R}F^2(r)w(r)\di r.
\end{equation*}
In the case of $\LSq(B)$ or $\LSq(\S^2)$ we have to change the integration domain and the integration measure accordingly to get the appropriate scalar product and norm, where the weight function simply becomes $w\equiv 1$ then. 
It is worth remembering that the $\LSq-$spaces technically 
contain equivalence classes of square-integrable functions, where two functions $F_1,F_2$
are equivalent if $\|F_1-F_2\|_{\LSq}=0$.
In the following chapter we will often use the \textbf{fully normalized spherical 
harmonics} $Y_{n,j}$ (where $n\in\N_0$ and $j\in\{-n,\dots,n\}$). These functions form
a complete orthonormal system of the Hilbert space $\LSq(\S^2)$ and are given by
\begin{equation}\label{eq:DefSPH}
    Y_{n,j}(\bi(\varphi,t)) := \sqrt{\frac{2n+1}{4\pi}\frac{(n-|j|)!}{(n+|j|)!}(2-\delta_{j0})}
    P_{n,|j|}(t)\begin{cases}
        \sin(j\varphi) & ,\text{ if }j>0 \\
        \cos(j\varphi) & ,\text{ if }j\leq 0
    \end{cases}
\end{equation}
in spherical coordinates. Here $n\in\N_0$ and $j\in\{-n,\dots,n\}$ and $P_{n,|j|}$ refers to the
well-known associated Legendre functions. The prefix fully normalized just means that we demand \begin{equation*}
    1 = \|Y_{n,j}\|_{\LSq(\S^2)}^2 = \int_{\S^2}Y_{n,j}^2(\bi)\di\omega(\bi)    
\end{equation*}
for every valid index $n$ and $j$.
If $j=0$, then we can simplify in Equation 
\eqref{eq:DefSPH}:
\begin{equation}\label{eq:DefSPHj0}
    Y_{n,0}(\bi(\varphi,t)) = \sqrt{\frac{2n+1}{4\pi}}P_n(t) = 
    \sqrt{\frac{2n+1}{4\pi}}P_n(\xi_3),
\end{equation} 
where the $P_n$ are the Legendre polynomials.

\section{The Gravitational Potential}\label{sec:GravPot}

We will assume that the deformation of the gas giant due to its own rotation is small 
enough such that we can model the planet as $B_R(0)=:B$, where $R>0$ is a characteristic radius (for example the equatorial radius). 
The point $\bm{0}\in B$ should lie in the \textbf{center of mass} $\bx_{\mathrm{C}}$, so 
\begin{align*}
    \bm{0} = \bx_{\mathrm{C}} =: \frac{1}{M}\int_{B}\bx \rho(\bx)\di\bx,
\end{align*}
where $M = \int_B\rho(\bx)\di\bx$ refers to the total planet mass and $\rho\in\Cont^2(\overline{B})$ is the mass density of the planet.
We can then also define the \textbf{gravitational potential} caused by $\rho$ as
\begin{align}
    V(\bx):=-G\int_{B}\frac{\rho(\bm{y})}{|\bx-\bm{y}|}\di \bm{y}\label{eq:GravPot}
\end{align}
for any $\bx\in\R^3.$ $G$ is the gravitational constant
($G\approx \qty{6.67e-11}{\metre^3\kilogram^{-1}\second^{-2}}$, see \cite{GravConst}). 
The gravitational potential behaves as the potential of a point mass at infinity, 
a property which is called \textbf{regularity at infinity}. 
Specifically this means that $|V(\bx)|=\mathcal{O}\left(|\bx|^{-1}\right)$ and 
$|\nabla V(\bx)| = \mathcal{O}\left(|\bx|^{-2}\right)$ holds for $|\bx|\to\infty$.
Here $\mathcal{O}$ is the \textbf{Landau symbol}, which is defined via 
\begin{align*}
    f(\bx) = \mathcal{O}(g(\bx))\text{ for }x\to\infty \,:\Leftrightarrow 
    \exists K,C>0:\,\forall |\bx|\geq C:\,|f(\bx)|\leq K|g(\bx)|.
\end{align*}
We now summarize some basic facts about the gravitational potential as defined
above (for proofs of these statements see e.g. the summary at the end of Chapter 
3.1 in \cite{Geomath}):
\begin{enumerate}[label={\Roman*.}]
    \item $V$ exists on all of $\R^3$, is bounded, continuous and partially
    differentiable everywhere. We can also exchange integration and
     differentiation.
    \item $V\in \Cont^2(\R^3\setm\overline{B})$ and 
    satisfies $\Delta V=0$ on $\R^3\setm\overline{B}$. $V$ is also regular at infinity.
    \item $V\in \Cont^2(B)$ and satisfies $\Delta V=4\pi G\rho$ on $B$.
\end{enumerate}
It is also worth noting that $V$ is continuously differentiable everywhere. Although
this is not proven in \cite*{Geomath}, a proof of this is easily adapted from Theorem 
3.1.2 in \cite*{Geomath}, using Corollary 3.1.4 in the same book.

\begin{assumptionS}\label{asmp:PhiIndependence}
    Planets are deformed by their rotation. 
    The choice of a spherical reference system where the (mean) rotation axis coincides with the z-axis ($\vartheta=0$) guarantees that the resulting density distribution and thus the gravity field is axisymmetric, i.e. independent of longitude $\varphi$, and symmetric with respect to the equatorial plane \cite*{JupiterGravFieldIess,SaturnGravFieldIess}. 
    Thus 
    \begin{equation*}
        \int_B|\bx|^nY_{n,j}\left(\frac{\bx}{|\bx|}\right)
        \rho(\bx)\di \bx=0
    \end{equation*}
    holds for $j\neq 0$.
\end{assumptionS}
We now have every tool to expand the gravitational potential on $\R^3\setm B$ 
into spherical harmonics. Explicitly, we get
\begin{equation}
    V(r\bi) = -\frac{GM}{r}\left(1-
    \sum_{n=2}^{\infty}J_n\left(\frac{\rad}{r}\right)^nP_n(\bi\cdot\bm{e}^3)
    \right)\label{eq:PotExpan}
\end{equation}
if $r>R$. The coefficients $J_n$ are defined as
\begin{equation*}
    J_n:=-\frac{1}{MR^n}\int_{B}|\bx|^nP_n\left(\frac{x_3}{|\bx|}\right)
    \rho(\bx)\di\bx\text{ for }n\geq 2.
\end{equation*}
The derivation of this equation uses the fact that 
\begin{equation*}
    \frac{1}{|\bx-\by|} = \sum_{n=0}^{\infty}\frac{|\bx|^n}{|\by|^{n+1}}
    P_n\left(\frac{\bx}{|\bx|}\cdot\frac{\by}{|\by|}\right)
\end{equation*}
for any two vectors $\bx,\by\in\R^3$ with $|\bx|<|\by|$ (see e.g. Corollary 3.4.18
in \cite{Geomath} for a proof of this). We also need that
$\frac{2n+1}{4\pi}P_n(\bi\cdot\bm{\eta}) = \sum_{j=-n}^nY_{n,j}(\bi)Y_{n,j}(\bm{\eta})$
holds for any two $\bi,\bm{\eta}\in\S^2$ (see e.g. Theorem 5.11 on page 103 in \cite{ConstrAppr}). 
With these two facts we see that 
\begin{align*}
    V(r\bi) &= -G\sum_{n=0}^{\infty}\sum_{j=-n}^n \frac{4\pi}{2n+1}r^{-n-1}Y_{n,j}(\bi)
    \int_B\rho(\by)|\by|^nY_{n,j}\left(\frac{\by}{|\by|}\right)\di\by \\
    &= -\frac{G}{r}\sum_{n=0}^{\infty}r^{-n}\sqrt{\frac{4\pi}{2n+1}}Y_{n,0}(\bi)
    \int_{0}^{R}\int_{\S^2}s^{n+2}\rho(s\bm{\eta})\sqrt{\frac{4\pi}{2n+1}}Y_{n,0}(\bm{\eta})
    \di\omega(\bm{\eta})\di s.
\end{align*}
Because $Y_{n,0}(\bi) = \sqrt{\frac{2n+1}{4\pi}}P_n(\bi\cdot\bm{e}^3)$, we can conclude
\begin{align*}
    & \int_{0}^{R}\int_{\S^2}s^{2}\rho(s\bm{\eta})\sqrt{\frac{4\pi}{2n+1}}Y_{0,0}(\bm{\eta})
    \di\omega(\bm{\eta})\di s
    =\int_B\rho(\by)\di\by = M, \\
    & \int_{0}^{R}\int_{\S^2}s^{3}\rho(s\bm{\eta})\sqrt{\frac{4\pi}{2n+1}}Y_{1,0}(\bm{\eta})
    \di\omega(\bm{\eta})\di s = \int_B y_3\rho(\by)\di\by = (\bx_C)_3 = 0\text{ and }\\
    & \int_{0}^{R}\int_{\S^2}s^{n+2}\rho(s\bm{\eta})\sqrt{\frac{4\pi}{2n+1}}Y_{n,0}(\bm{\eta})
    \di\omega(\bm{\eta})\di s = -MR^nJ_n\text{ for all }n\geq 2.
\end{align*}
Therefore Equation \eqref{eq:PotExpan} holds.
This expansion is common in the literature 
\cite{GravSignZonalFlowsHubbard,CylindWindsKaspi,ZonalWindsWicht,TGWEZhang}
and the coefficients $J_n$, calculated from the observed data, are available for 
computations.
See for example \cite{AtmosWindsKaspi}, for a table containing the coefficients $J_2$ to $J_{10}$ for Jupiter as well as Saturn.
More recently the coefficients have also been published up to $J_{40}$ (see \cite{CylindWindsKaspi}), but the method of deriving the higher order coefficients is based on some additional assumptions on the zonal wind fields.

\section{Modeling the Wind inside the Planet}\label{sec:WindInsidePlanet}
The scheme we will use to solve the forward problem of calculating the gravitational coefficients from the zonal winds is summarized in Figure \ref{fig:FlowchartFWP}.
\begin{figure}[ht]
    \centering
    \includegraphics[width=0.7\textwidth]{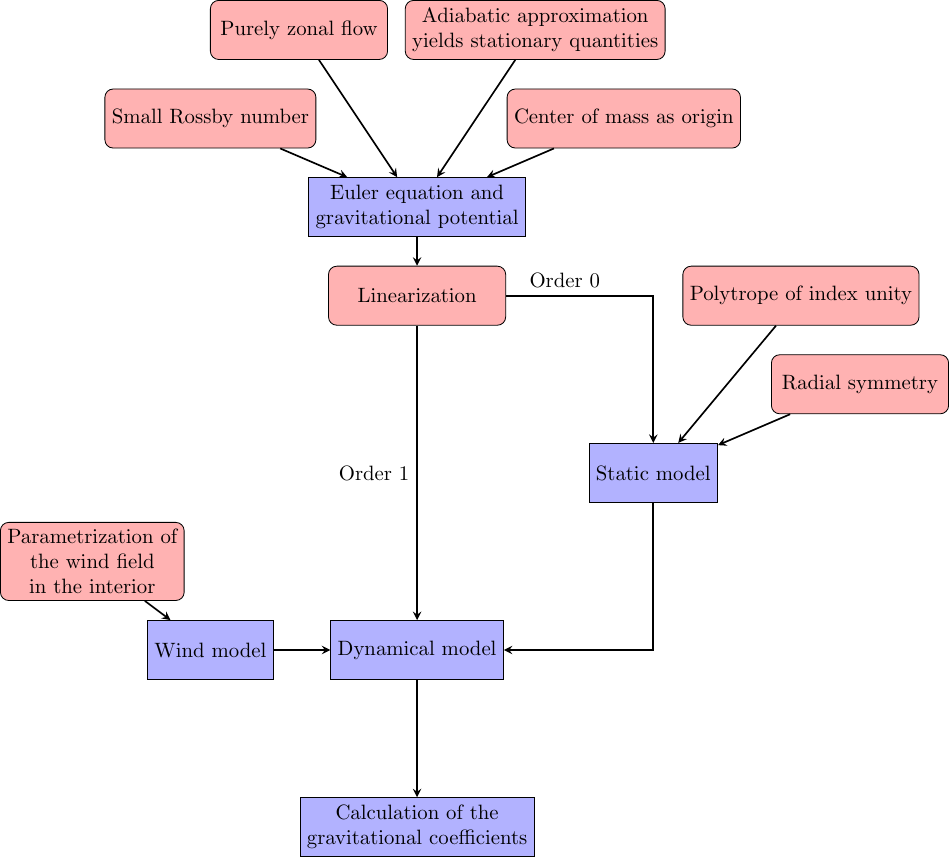}
    \caption{Flowchart of the forward problem}\label{fig:FlowchartFWP}
\end{figure}
Early models of the Jovian atmosphere used the \textbf{Euler equation} in a non-rotating
coordinate system as a starting point \cite{GravSignZonalFlowsHubbard}, but today 
the Euler equation in a rotating reference frame is often used 
\cite{DeepWindStructKaspi,AtmosWindsKaspi,ZonalWindsWicht,TGWEZhang}. 
Of course we will also need \textbf{mass conservation}.
\begin{align}\label{eq:Euler}
    \partial_t\bm{u}+\left(\bm{u}\cdot\nabla\right)\bm{u}+2\bm{\Omega}
    \times\bm{u}+\bm{\Omega}\times\left(\bm{\Omega}\times\bx\right) & =
    -\frac{1}{\rho}\nabla p+\bm{g}\text{ for }\bx\in B, \\ \label{eq:MassCon}
    \partial_t\rho+\nabla\cdot(\rho\bm{u}) & =0\text{ for }\bx\in B.
\end{align}
Here $\bm{u}$ is the wind velocity, $\bm{\Omega} = \Omega\bm{e}^3$ is the constant angular
rotation rate of the planet, $\rho$ is the mass density from Chapter \ref{sec:GravPot} 
satisfying Assumption \ref{asmp:PhiIndependence}, $p$ is the pressure and 
$\bm{g}=-\nabla V$ is the gravitational acceleration caused by the potential $V$ from 
Chapter \ref{sec:GravPot}. These quantities will be assumed to be twice continuously 
differentiable on $B$. $\bx\in\R^3$ represents location and for now
$t\geq 0$ represents the time, although we will later neglect the time argument. 
$B$ is defined as in the previous section.
It is worth noting that we know the wind field on the surface of the planet (see for 
example \cite{SurfaceWindSaturn,SurfaceWindJupiter}). The Euler equation is then simplified
further using the following asssumptions \cite{SelfConsistentTreatmentKaspi}.
\begin{assumptionS}\label{asmp:TimeIndependence}
    Experiments, numerical simulations, and observations show that the flows in gas planets are dominated by zonal flows.
    Since these flows are rather constant in time, we can neglect the time derivative in the Euler equation \eqref{eq:Euler}.
    Local variations in density can also be neglected in the so-called anelastic approximation, which eliminated sound waves. We can thus ignore the time derivative in Equation \eqref{eq:MassCon} and drop the time-dependence from the problem.
    From now on we will omit the time argument in all physical quantities and $t$ will always refer to the polar distance in spherical coordinates.
\end{assumptionS}
\begin{assumptionS}\label{asmp:SmallRossby}
    The zonal flows are much slower than the rotation speed $R\Omega$ of the planet. The Rossby number, providing the ratio of flow speed to rotation speed, is smaller than 0.01 for Jupiter and smaller than 0.05 for Saturn. 
    We can thus neglect the non-linear inertial term $(\bm{u}\cdot\nabla)\bm{u}$ in equation \eqref{eq:Euler}.
\end{assumptionS}
\begin{remarkS}
    Note that we also neglect magnetic effects, since the conductivity in the outer planetary layers we are interested in is likely too low to support significant Lorentz forces.
\end{remarkS}
We can further simplify Equation \eqref{eq:Euler} by writing the centrifugal acceleration as a gradient field. The \textbf{centrifugal potential} $V_{\Omega}$ 
satisfies $\nabla V_{\Omega}(\bx)=\bm{\Omega}\times(\bm{\Omega}\times\bx)$, therefore 
Equation \eqref{eq:Euler} becomes
\begin{equation}\label{eq:EulerSimp}
    2\bm{\Omega}\times(\rho\bm{u})= -\nabla p-\rho\nabla(V+V_{\Omega}).
\end{equation}
Taking the curl and remembering that the curl of a gradient field vanishes rule yields the following on the right-hand side of Equation 
\eqref{eq:EulerSimp}
\begin{equation*}
    \nabla\times(-\nabla p-\rho\nabla(V+V_{\Omega})) = 
    -\nabla \rho\times\nabla(V+V_{\Omega}).
\end{equation*}
On the left-hand side of Equation \eqref{eq:EulerSimp} we get
\begin{equation*}
    \nabla\times(\bm{\Omega}\times(\rho\bm{u})) = \bm{\Omega}\nabla\cdot(\rho\bm{u})
    -(\bm{\Omega}\cdot\nabla)(\rho\bm{u}) = -(\bm{\Omega}\cdot\nabla)(\rho\bm{u}),
\end{equation*}
where the last equality is due to Equation \eqref{eq:MassCon} in the 
anelastic approximation. Thus the following partial differential equation holds on 
$B$:
\begin{equation}
    2(\bm{\Omega}\cdot\nabla)(\rho\bm{u}) =\nabla \rho\times\nabla
    (V+V_{\Omega})\label{eq:UnperThermalWind}.
\end{equation}
This is still too complicated, because 
the right-hand side is nonlinear in $\rho$ ($V$ also depends on $\rho$). 
To have a chance to solve this equation, we need to linearize it around some background state.

\subsection{A Perturbation Approach}
In this chapter we will use the standard method to linearize Equation 
\eqref{eq:UnperThermalWind}, laid out in 
\cite{SelfConsistentTreatmentKaspi,DeepWindStructKaspi,OriginZonalWindKong,ZonalWindsWicht} as well as 
\cite{TGWEZhang}.
\begin{definitionS}\label{def:StatQuant}
    The force balance, expressed by the Euler equation \eqref{eq:Euler} or the adapted version \eqref{eq:UnperThermalWind}, is by far dominated by the hydrostatic balance between pressure gradient and (modified) gravity potential. This defines a \textbf{stationary background state} $p_0,\rho_0, V_0$ without flow on $B$ via
    \begin{equation}\label{eq:HydrEquib}
        0 = -\nabla p_0-\rho_0\nabla(V_0+V_{\Omega}).
    \end{equation}
\end{definitionS}
\begin{definitionS}\label{def:DynQuant}
    We define the \textbf{dynamical quantities} or perturbations $\rho'$, $V'$ and $p'$ as 
    the disturbances caused by the zonal flow
    \begin{align*}
        \bm{u}'(\bx) & :=\bm{u}(\bx)-\bm{u}_0(\bx)=\bm{u}(\bx),\text{ so } \\
        \rho'(\bx) & :=\rho(\bx)-\rho_0(\bx), \\
        V'(\bx) & :=V(\bx)-V_0(\bx)\text{ and } \\
        p'(\bx) & :=p(\bx)-p_0(\bx).
    \end{align*}
    Here $\rho$, $V$ and $p$ satisfy Equation \eqref{eq:UnperThermalWind}.
    The interpretation is that $\rho',V'$ and $p'$ are purely caused
    by the wind field $\bm{u}$. Because of the linearity of the gravitational
    potential, $V'$ and $\rho'$ satisfy the relation in equation \eqref{eq:GravPot}.
\end{definitionS}
Plugging this into Equation \eqref{eq:UnperThermalWind} and neglecting the terms of 
second order in the dynamic quantities (this includes combinations of primed quantities 
and $\bm{u}$, for example the term \newline$2(\bm{\Omega}\cdot\nabla)(\rho'\bm{u})$) yields
\begin{equation}\label{eq:ThermalWindPer}
    2(\bm{\Omega}\cdot\nabla)(\rho_0\bm{u})= \nabla \rho'
    \times\nabla(V_0+V_{\Omega})-\nabla V'\times\nabla \rho_0\text{ on }B.
\end{equation}
We have used $0 = \nabla \rho_0\times\nabla(V_0+V_{\Omega})$, which is the curl of Equation
\eqref{eq:HydrEquib}. 
Having now split up the problem into two smaller problems, we first 
determine the quantities $\rho_0,V_0$ and $p_0$ from the equation in Definition 
\ref{def:StatQuant}. 
For this an appropiate model for a rigidly rotating gas giant is 
needed. 
Having determined the hydrostatic background quantities, we can plug these into the equations in 
Equation \eqref{eq:ThermalWindPer} and try to determine $\rho'$ from $\bm{u}$.

\subsection{The Static Problem}\label{sec:StatModel}
Many models have been considered to calculate the interior structure of a rigidly rotating
gas giant. There is the theory of figures from \cite{ZharkovTrubitsyn}, further developed in
\cite{ConvToFHubbard,GravHarmRigidJupiterNettelmann,ToFSeventhOrderNettelmann}, which 
determines the shape and the gravitational potential of a gas giant by utilizing an expansion 
in the small parameter $q = \Omega^2R^3/(GM)$. A second method, developed in \cite{CMSHubbard}
takes advantage of the fact that the gravitational potential of a Maclaurin spheroid with 
constant density can be calculated analytically. A general model of the gas giants interior is 
then constructed by layering concentric Maclaurin spheroids with constant density in a 
self-consistent manner (this is called the CMS-model). See 
\cite{CMSHubbardBasics,CMSHubbard,CompLayerModelsHubbard,DiffRotCompWisdom} for further 
explanation and usage of this model.
A third model which is worth mentioning is explained in 
\cite{GravSignZonalFlowsHubbard,NonPertubEquilbWisdom,DiffRotCompWisdom} and uses the 
assumption of a \textbf{polytrope of index unity} to derive a solution in terms of 
spherical Bessel functions. 
This assumption has been studied for a long time \cite{GravFieldPolytropeHubbard} and we will use it to solve for the static solutions $p_0,\rho_0\in \Cont^2(\overline{B})$ and $V_0$ of Equation \eqref{eq:HydrEquib}.
A somewhat more general approach for an adiabatic background state can be found in \cite{ZonalWindsWicht}
\begin{assumptionS}\label{asmp:PolytropeStatModel}
    In the following we will assume that the static density and pressure
    obey a \textbf{polytrope of index unity}, which means that there exists some 
    constant $K>0$ such that
    \begin{align*}
        p_0=K\rho_0^2\text{ on }B.
    \end{align*}
    In general a polytrope of index $n\in\N$ refers to the relation
    $p_0=K\rho_0^{1+1/n}$.
\end{assumptionS}
It is easy to see that under Assumption \ref*{asmp:PolytropeStatModel}, we can 
calculate the gradient of the static potential, due to Equation \eqref{eq:HydrEquib},
by the formula
\begin{equation}
    \nabla (V_0+V_{\Omega}) = -2K\nabla \rho_0\text{ on }B.\label{eq:PotGradStatic}
\end{equation}
Taking the divergence of Equation \eqref{eq:PotGradStatic} yields
\begin{alignat}{3}\nonumber
    & \Delta V_0+\Delta V_{\Omega} &&= -2K\Delta \rho_0 \\\nonumber
    \Leftrightarrow & (4\pi G+2K\Delta)\rho_0 &&= -\Delta V_{\Omega} \\\label{eq:HelmholtzDensity}
    \Leftrightarrow & \left(\Delta+\frac{2\pi G}{K}\right)\rho_0 &&= \frac{\Omega^2}{K}
\end{alignat}
on $B$, where we have used Poisson's equation $\Delta V_0 = 4\pi G\rho_0$ and 
$\Delta V_{\Omega} = -2\Omega^2$.
Following \cite{GravSignatureGasGiantsKaspi,ZonalWindsWicht} we will next assume 
spherical symmetry.
\begin{assumptionS}\label{asmp:RadIndependence}
    The mass density $\rho_0$ only depends on $|\bx|$ or $r$ in spherical 
    coordinates. This assumption is justified because the centrifugal potential $V_{\Omega}$
    is small compared to the radially symmetric background potential $V_0$. In particular, this means that Equation \eqref{eq:HelmholtzDensity} becomes homogeneous.
\end{assumptionS}
Solving Equation \eqref{eq:HelmholtzDensity} in the case of radial symmetry is relatively straight-forward:
We only need the radial part $\partial_r^2+\frac{2}{r}\partial_r$ of the Laplacian $\Delta$, multiplying 
Equation \eqref{eq:HelmholtzDensity} by $r^2$ then yields
\begin{equation*}
    r^2\frac{\di^2}{\di r^2}\rho_0(r)+2r\frac{\di}{\di r}\rho_0(r)+\frac{2\pi G}{K}r^2\rho_0(r) = 0
\end{equation*}
for every $r\in[0,R]$. Substituting $\gamma := \sqrt{2\pi G/K}$ and $s := \gamma r$ as well 
as $G(s) := \rho_0(r)$ in the above equation gets us
\begin{equation}
    s^2\frac{\di^2}{\di s^2}G(s)+2s\frac{\di}{\di s}G(s)+s^2G(s) = 0.
\end{equation}
The two linearly independent solutions of this equation are the spherical Bessel functions 
of degree $0$, $j_0$ and $y_0$ (see Appendix 
\ref{sec:BesselProp}), so $\rho_0$ is given by 
\begin{equation*}
    \rho_0(r) = cj_0(\gamma r)+dy_0(\gamma r).
\end{equation*}
Since we want our density to be regular at $r=0$ and $y_0$ has a singularity at this spot, 
we conclude $d=0$. So far we deduced the following expression for $\rho_0$:

\begin{equation*}
    \rho_0(r) = cj_0\left(\sqrt{\frac{2\pi G}{K}}r\right),
\end{equation*}
where $c\in\R$ is an arbitrary constant. Since $\rho_0$ is supposed to be the mass density 
of a gas giant, we demand $\rho_0\geq 0$ on $[0,R]$ and $\rho_0(R)=0$. These conditions lead us to 
$\gamma = \pi/R$ or $K=\frac{2GR^2}{\pi}$. Finally we normalize $\rho_0$ such that 
$\int_B\rho_0(\bx)\di\bx=M$, meaning that we can determine $c$. Using the substitution
$s = \pi r/R$ we conclude that
\begin{align*}
    \int_Bj_0\left(\frac{\pi}{R}|\bx|\right)\di\bx &= 4\pi\int_{0}^{R}r^2j_0
    \left(\frac{\pi}{R}r\right)\di r = 4\pi\int_{0}^{\pi}\left(\frac{R}{\pi}s\right)^2
    j_0(s)\frac{R}{\pi}\di s \\ 
    &= \frac{4R^3}{\pi^2}\int_{0}^{\pi}s\sin(s)\di s = \frac{4R^3}{\pi}.
\end{align*}
Therefore the constant $c$ must have the value $\pi M/(4R^3)$ and our density looks 
like
\begin{equation*}
    \rho_0(r) = \overline{\rho_0}\,\frac{\pi^2}{3}j_0\left(\pi\frac{r}{R}\right),
\end{equation*}
where $\overline{\rho_0}=3M/(4\pi R^3)$ is the mean density.

\subsection{The Dynamic Problem}\label{sec:TGWEq}
To solve the dynamic problem there are essentially two approaches. The first, considered by 
\cite{CylindWindsKaspi,kaspiGravitationalSignatureJupiters2010} for example makes the additional
assumption that the gravitational acceleration caused by the wind field will be negligible, meaning
that
\begin{equation}\label{eq:TWAssump}
    \nabla V'(\bx) \approx 0\quad\text{for every }\bx\in B.
\end{equation}
With this assumption Equation \eqref{eq:ThermalWindPer} becomes simpler again, which yields 
 the \textbf{Thermal Wind Equation}
\begin{equation}\label{eq:ThermalWind}
    2(\bm{\Omega}\cdot\nabla)(\rho_0\bm{u})= \nabla \rho'
    \times\nabla(V_0+V_{\Omega})\quad\text{on }B.
\end{equation}
We now use Equation \eqref{eq:PotGradStatic} and then Assumption \ref{asmp:RadIndependence},
which implies that
\begin{align*}
    2(\bm{\Omega}\cdot\nabla)(\rho_0\bm{u}) & =
    \nabla \rho_0\times\nabla(2K\rho') \\
    \Leftrightarrow 2(\bm{\Omega}\cdot\nabla)(\rho_0\bm{u}) & =
    \frac{\partial_r \rho_0}{r}\Lup^{\ast}(2K\rho').
\end{align*}
It is easy to see that in the case of $\varphi$-independence, $\Lup^{\ast}$ 
only has a component in the $\beps^{\varphi}$-direction, so we take the inner product 
with this vector on both sides. With the definition $u_{\varphi}:=\bm{u}\cdot
\bm{e}^{\varphi}$ we arrive at the equation
\begin{equation}\label{eq:TWSimplestSolStep1}
    2(\bm{\Omega}\cdot\nabla)(\rho_0u_{\varphi}) = 
    \frac{\partial_r\rho_0}{r}\left(\bm{e}^{\varphi}\cdot\Lup^{\ast}\right)
    (2K\rho')\text{ on }B,
\end{equation}
because $\bm{\Omega}\cdot\nabla$ only involves derivatives 
in the $x_3$-direction. At this point it is useful to introduce the following
quantity.
\begin{definitionS}\label{def:SourceTGW}
    We define the \textbf{thermal-wind source term} $S^u$ as the function 
    \begin{align*}
        S^u(r\bi(\varphi,t)) := -\frac{4\pi G}{K}\,\frac{r}{\partial_r
        \rho_0(r\bi(\varphi,t))}\int_{-1}^{t}\frac{(\bm{\Omega}\cdot\nabla)(\rho_0u_{\varphi})
        (r\bi(\varphi,\tau))}{\sqrt{1-\tau^2}}\di\tau,
    \end{align*}
    where $r\in[0,R]$, $\varphi\in[0,2\pi)$ and $t\in[-1,1]$ are spherical coordinates.
    The quantities $\rho^0$ and $u_{\varphi}$ are by assumption independent of the parameter 
    $\varphi$ in spherical coordinates. Thus the quantity $S^u$ also satisfies $\partial_{\varphi}S^u=0$, just as $\rho^0$ and $u_{\varphi}$ do.
\end{definitionS}
We notice that $(\beps^{\varphi}\cdot\Lup^{\ast})Q = -\sqrt{1-t^2}\partial_tQ$ and conclude 
that Equation \eqref{eq:TWSimplestSolStep1} yields after a rearrangement the identity 
\begin{equation*}
    -\frac{2(\bm{\Omega}\cdot\nabla)(\rho_0u_{\varphi})(r\bi(\varphi,t))}{\sqrt{1-t^2}}\cdot
    \frac{r}{\partial_r\rho_0(r\bi(\varphi,t))} = \partial_t\left(2K\rho'(r\bi(\varphi,t))\right)
\end{equation*}
for all $(r,\varphi,t)$. Now, we integrate both sides with respect to $t$, then we get a 
constant of integration $c$ which only depends on $r$. With Definition \ref*{def:SourceTGW} 
this results in
\begin{align*}
    \frac{K}{2\pi G}S^u(r\bi(\varphi,t))+c(r) = 2K\rho'(r\bi(\varphi,t)).
\end{align*}
Since $\Delta V' = 4\pi G \rho'$ is satisfied on $B$ we get
\begin{alignat}{3}\nonumber
    && 4\pi G \rho'(r\bi) &= S^u(r\bi)+\eta(r)\\ \label{eq:PoissonKaspi}
    \Leftrightarrow && \Delta V'(r\bi) &= S^u(r\bi)+\eta(r),
\end{alignat}
where $\eta(r):=2\pi Gc(r)/K $. With this equation we can determine $\rho'$ (or equivalently $\Delta V'$)
up to an additive function only dependent on the radius, if we know the wind field 
$u_{\varphi}$. As we will see later, this is enough to fully determine the gravitational field 
outside of the planet.

We will also highlight a second method to solve the dynamic problem, which was introduced by 
\cite{TGWEZhang} and further developed by \cite{ZonalWindsWicht}. This approach does not 
neglect the term $\nabla V'$, thus we start out with Equation \eqref{eq:ThermalWindPer}, which 
is called the \textbf{Thermo-Gravitational Wind Equation} in this context. Proceeding in exactly the 
same fashion as with Equation \eqref{eq:ThermalWind} (applying Equation \eqref{eq:PotGradStatic}
and then Assumption \ref{asmp:RadIndependence}) we arrive at
\begin{equation*}
    2(\bm{\Omega}\cdot\nabla)(\rho_0\bm{u}) =
    \frac{\partial_r \rho_0}{r}\Lup^{\ast}(2K\rho'+V')\quad\text{ on }B.
\end{equation*}
This is almost the same equation as Equation \eqref{eq:TWSimplestSolStep1}, but we have replaced 
$2K\rho'$ with $2K\rho'+V'$. Thus we can transform our equation in exactly the same way as above, yielding
\begin{equation*}
    \frac{K}{2\pi G}S^u(r\bi(\varphi,t))+c(r) = 2K\rho'(r\bi(\varphi,t))+
    V'(r\bi(\varphi,t))
\end{equation*}
with an indeterminate function $c(r)$ which only depends on the radius $r$.
Finally, using that $\Delta V' = 4\pi G \rho'$ is satisfied on $B$ and 
$2\pi G/K = \pi^2/R^2$ we obtain the inhomogeneous Helmholtz equation
\begin{equation}\label{eq:HelmholtzWicht}
    \left(\Delta+\frac{\pi^2}{R^2}\right) V'(r\bi) = 
    S^u(r\bi)+\eta(r),
\end{equation}
where $\eta(r):=2\pi Gc(r)/K $. To determine $\rho'$ or $\Delta V'$ from this equation,
a suitable orthonormal basis is necessary, which was first introduced by 
\cite{ZonalWindsWicht}.

\section{An Orthonormal Basis for the Gravitational Potential}\label{sec:ONBGravPot}
We now construct an orthonormal basis of $\LSq(B)$ which is useful for expanding 
the gravitational potential in $B$, as is done in \cite{ZonalWindsWicht}. 
Of course we know that the gravitational potential 
$V$ is in $\Cont^1(\R^3)$, satisfies $\Delta V = 0$ on $\R^3\setminus\overline{B}$ and 
is regular at infinity. It is well-known (see for example \cite*{Geomath}) that these 
conditions fully determine $V$ on $\R^3\setminus\overline{B}$, namely
\begin{equation}\label{eq:ExtDPSol}
    V(r\bi)=\sum_{n=0}^{\infty}\sum_{j=-n}^{n}\langle V(R\,\cdot),Y_{n,j}
    \rangle_{\LSq(\S^2)}\left(\frac{R}{r}\right)^{n+1}Y_{n,j}(\bi)
\end{equation}
holds for every $r>R$. At this point a note on notation is necessary. For a function $f:\R^3\to\R$ and a fixed $s>0$, the function $f(s\cdot)$ is defined via 
\begin{align*}
    f(s\cdot):\S^2\to\R,\, \bi\mapsto f(s\bi).    
\end{align*}
The dot in the argument of $f$ symbolizes that we keep the radial coordinate $s$ fixed here and we only consider the dependence of the function on the direction $\bi$.
Equation \eqref{eq:ExtDPSol} allows us to determine the potential on $\partial B$,
because in the sense of $\LSq(\S^2)$ we have 
\begin{align}\nonumber
    V(R\bi) &= \sum_{n=0}^{\infty}\sum_{j=-n}^{n}\langle V(R\,\cdot),
    Y_{n,j}\rangle_{\LSq(\S^2)}Y_{n,j}(\bi)\quad\text{and}\\\label{eq:BoundCondPot}
    \partial_rV(r\bi)\big|_{r=R} &= -\sum_{n=0}^{\infty}\sum_{j=-n}^{n}\frac{n+1}{R}
    \langle V(R\,\cdot),Y_{n,j}\rangle_{\LSq(\S^2)}Y_{n,j}(\bi).
\end{align}
These formulas will motivate the boundary conditions which our basis should satisfy.
The structure of our basis will follow the decomposition of the inner product on $\LSq(B)$
into an angular and a radial part. An arbitrary basis function $u\in\LSq(B)$ should have 
the form 
\begin{equation*}
    u(\bx) = u(r\bi) = H(r)Y_{n,j}(\bi)
\end{equation*}
for every $\bx\neq 0$. Here the $Y_{n,j}$ are the fully normalized spherical harmonics, 
which form an orthonormal basis of $\LSq(\S^2)$. It is then easy to see that our functions
$H$ should form an orthonormal basis of $\LSq_w(0,R)$, where $w(r):=r^2$.
Since we want some similarity to the gravitational potential, we 
demand a boundary condition motivated by Equation \eqref{eq:BoundCondPot} and the 
regularity of $V$ in the interior of $B$: the function $H$ which forms the radial part of 
one of our basis functions $u$ satisfies $H\in\Cont^2[0,R]$ and $H'(R) = 
-(n+1)H(R)/R$, where the prime $'$ just denotes the derivative of $H$ with respect to the variable $r$. Since we often encounter Laplacians when dealing with potentials,
we also demand that our basis functions are eigenfunctions of the Laplacian on $B$. 
This means any basis function $u$ should satisfy
\begin{align}\label{eq:EFLaplace}
    \Delta u = -\lambda u\quad\text{on } B.
\end{align}
At first we only demand $\lambda\in\C$ but we will restrict this later. Using the 
property $\Delta^{\ast}Y_{n,j} = -n(n+1)Y_{n,j}$ and the decomposition of the Laplacian
\begin{equation*}
    \Delta = \partial_r^2+\frac{2}{r}\partial_r+\frac{1}{r^2}\Delta^{\ast}
\end{equation*}
on Equation \eqref{eq:EFLaplace} yields the following for the radial part $H\in\Cont^2[0,R]$ 
of our basis function
\begin{alignat}{3}\nonumber
    && H''(r)Y_{n,j}(\bi)+\frac{2}{r}H'(r)Y_{n,j}(\bi)-\frac{n(n+1)}{r^2}H(r)Y_{n,j}(\bi) &= 
    -\lambda H(r)Y_{n,j}(\bi) \\ \nonumber
    \Leftrightarrow && \frac{1}{r^2}\left(\frac{\di}{\di r}\left(r^2\frac{\di}{\di r}
    H(r)\right)-n(n+1)H(r)\right)Y_{n,j}(\bi) &= -\lambda H(r)Y_{n,j}(\bi) \\
    \label{eq:SLOpRad}
    \Leftrightarrow && \Lup H(r) := \frac{\di}{\di r}\left(r^2\frac{\di}{\di r}
    H(r)\right)-n(n+1)H(r) &= -\lambda w(r)H(r),
\end{alignat}
which holds for every $r\in(0,R)$. In the last step we have used the fact that the 
$Y_{n,j}$ only have isolated zeros on $\S^2$, meaning we can factor these functions 
out of the equation. We also demand the boundary condition
\begin{equation}\label{eq:BC}
    H'(R) +\frac{n+1}{R}H(R) = 0.
\end{equation}
We will call $\lambda$ the eigenvalue of the problem and $H$ the eigenfunction.
Of course we get a different boundary value problem for every $n\in\N_0$, but we 
will suppress this in the notation for now.\\
Now we need to find every eigenvalue and -function of the problem defined above. 
First we will prove some general properties of the eigenvalues $\lambda$.
\begin{lemmaS}
    If we define the vector space 
\begin{equation*}
    V := \{h\in\Cont^2[0,R] \,|\, h\text{ satisfies \eqref{eq:BC}}\}
\end{equation*}
together with the inner product 
\begin{equation*}
    \langle F,G\rangle_V := \int_{0}^{R}F(r)G(r)\di r
\end{equation*}
for two functions $F,G\in V$, then the operator $\Lup$ defined in Equation 
\eqref{eq:SLOpRad} is self-adjoint and negative semi-definite on 
$(V,\langle\cdot,\rangle_V)$.
\end{lemmaS}
\begin{proof}[Proof]
    For the self-adjointness we prove that $\langle \Lup F,G\rangle_V=\langle F,
    \Lup G\rangle_V$ for any two functions $F,G\in V$. This involves integrating by parts
    twice. For the negative-semi definite claim, we prove that 
    $\langle\Lup F,F\rangle_V\leq 0$ for any $F\in V$, integrating by parts once.
    In both cases we have to use Equation \eqref{eq:BC} on the boundary terms resulting 
    from integration by parts.
\end{proof}
\begin{remarkS}\label{rem:OGEF}
    As it is well-known, self-adjoint operators only have real eigenvalues, which, in 
    the case of negative-semi-definiteness, need to be non-positive. Therefore we conclude
    that our numbers $\lambda$ solving Equation \eqref{eq:EFLaplace} must satisfy $\lambda\geq 0$ and will write $\lambda=\gamma^2$
    for some $\gamma\geq 0$ from now on. Finally we note that the eigenfunctions corresponding
    to different eigenvalues will be orthogonal.
\end{remarkS}
\begin{lemmaS}\label{lma:FormulaEF}
    For every arbitrary but fixed $n\in\N_0$, the $\Cont^2[0,R]$-solutions to Equation 
    \eqref{eq:SLOpRad} with boundary condition \eqref{eq:BC}, which are normalized in 
    $\LSq_w(0,R)$, are 
    \begin{equation*}
        B_{m,n}(r) := \pm\sqrt{\frac{2}{R^3j_n^2(\gamma_{n,m}R)}}j_n(\gamma_{n,m}r),
    \end{equation*}
    where the $(\gamma_{n,m})_{m\in\N}$ are the distinct positive solutions to 
    \begin{equation*}
        0 = j_{n-1}(\gamma R).
    \end{equation*}
    Here the $j_n$ are the spherical Bessel functions of the first kind, as defined in 
    Definition \ref{def:BesselFunc}.
\end{lemmaS}
\begin{proof}[Proof]
    We substitute $\lambda=\gamma^2$ for some $\gamma\geq 0$, then equation 
    \eqref{eq:SLOpRad} is equivalent to 
    \begin{equation*}
        r^2H''(r)+2rH'(r)+(\gamma^2r^2-n(n+1))H(r) = 0\quad\text{for every }r\in(0,R).
    \end{equation*}
    For $\gamma=0$, we have a Cauchy-Euler-equation with general solution of the form 
    $H(r) = b_1r^n+b_2r^{-n-1}$. This leaves only the trivial solution since regularity
    at $r=0$ yields $b_2=0$ and Equation \eqref{eq:BC} yields $b_1=0$.
    For $\gamma>0$, the general solution is 
    \begin{equation*}
        H(r) = c_1j_n(\gamma r)+c_2y_n(\gamma r),
    \end{equation*}
    where the $j_n$ and $y_n$ are the spherical Bessel functions of the first and second 
    kind. As our functions should be defined at $r=0$, while the functions $y_n$ are not, 
    we must conclude that $c_2=0$. To determine $\gamma$, we use the recurrence relation 
    found in 10.1.21 in \cite*{AaS} 
    \begin{align*}
        j_{n-1}(x) = j_n'(x)+\frac{n+1}{x}j_n(x)\quad\text{for every }n\in\N_0
        \text{ and }x>0.
    \end{align*}
    Then $\gamma$ must satisfy $j_{n-1}(\gamma R)=0$. It turns out that there are a 
    countable number of distinct positive solutions $(\gamma_{n,m})_{m\in\N}$ to this 
    equation (this follows from 9.5.2 in \cite*{AaS} as well as 15.21, 15.22 and the 
    fifth paragraph in 15.28 in \cite*{BesselWatson}). 
    To normalize the solution $H$ in $\LSq_w(0,R)$, we must choose $c_1$ such that 
    \begin{equation*}
        c_1 = \pm\left(\int_{0}^{R}r^2j_n^2(\gamma_{n,m}r)\di r\right)^{-1/2}.
    \end{equation*}
    With the definition $\lambda_{n,m}:=\gamma_{n,m}R$ and the substitution
    $s=r/R$ we conclude that 
    \begin{equation*}
        \frac{1}{c_1^2} = R^3\int_{0}^{1}s^2j_n^2(\lambda_{n,m}s)\di s,
    \end{equation*}
    where the $(\lambda_{n,m})_{m\in\N}$ are all the positive solutions of 
    $j_{n-1}(\lambda)=0$. Using Definition \ref{def:BesselFunc} we see that the 
    coefficients $\lambda_{n,m}$ comprise every positive solution of 
    $J_{n-1/2}(\lambda) = 0$ or equivalently with Lemma \ref{lma:RecurrenceBessel}
    applied for $\nu = n+1/2$
    \begin{equation*}
        \left(n+\frac{1}{2}\right)J_{n+1/2}(\lambda)+\lambda J_{n+1/2}'(\lambda)=0.
    \end{equation*}
    Using Definition \ref{def:BesselFunc} again and subsequently Theorem 
    \ref{thm:NormalizedBessel} leads us to
    \begin{equation*}
        \frac{1}{c_1^2} = \frac{\pi R^3}{2\lambda_{n,m}}\int_{0}^{1}s
        J_{n+1/2}^2(\lambda_{n,m}s)\di s = 
        \frac{R^3}{2}\,\frac{\pi}{2\lambda_{n,m}}J_{n+1/2}^2(\lambda_{n,m}) =
        \frac{R^3}{2}j_n^2(\gamma_{n,m}R).\qedhere
    \end{equation*}
\end{proof}
\begin{theoremS}\label{thm:ONBRadial}
    For a fixed but arbitrary $n\in\N_0$ and with the definitions from 
    above, the sequence of functions $(B_{m,n})_{m\in\N}$ -- where we always choose the
    $+$--variant from Lemma \ref{lma:FormulaEF} -- forms an orthonormal basis of 
    $\LSq_w(0,R)$.
\end{theoremS}
\begin{proof}[Proof]
    We know that the functions $B_{m,n}$ are eigenfunctions of the operator $\Lup$
    from Equation \eqref{eq:SLOpRad} for eigenvalues $\gamma_{n,m}^2$, which are all 
    distinct. This proves the orthogonality in $\LSq_w(0,R)$ by Remark \ref{rem:OGEF}. 
    The sequence is also normalized in this space as we have seen in Lemma 
    \ref{lma:FormulaEF}. What remains to be proven is the completeness in $\LSq_w(0,R)$. 
    To do this, we first take an arbitrary $H\in\Cont^2_{\textup{c}}(0,R)\subs 
    \LSq_w(0,R)$ (functions which are twice continuously differentiable and have compact 
    support in $(0,R)$) and show that the partial Fourier sum
    \begin{align*}
        H_M(r):=\sum_{m=1}^{M}\langle H,B_{m,n}\rangle_{\LSq_w(0,R)}B_{m,n}(r)
    \end{align*}
    converges to $H$ in $\LSq_w(0,R)$ as $M\to\infty$. This will happen in two steps. 
    First we show the pointwise convergence on $(0,R)$ and then the convergence in $\LSq$
    is achieved by proving the sequence $(H_M)_{M\in\N}$ is Cauchy.
    For both cases it is useful to transform the interval $[0,R]$ into $[0,1]$ by defining
    $h(s) := H(Rs)$ and $h_M(s) := H_M(Rs)$ for every $(r/R=)s\in[0,1]$. We also 
    substitute $\lambda_{n,m}:=\gamma_{n,m}R$. To make use of Appendix 
    \ref{sec:BesselProp} we also transform from spherical to regular Bessel functions (see again 
    Definition \ref{def:BesselFunc}). The $\lambda_{n,m}$ are the positive solutions to 
    $J_{n-1/2}(\lambda)=0$, which is equivalent to 
    \begin{equation*}
        \left(n+\frac{1}{2}\right)J_{n+1/2}(\lambda)+xJ_{n+1/2}'(\lambda) = 0
    \end{equation*}
    by recurrence relation \eqref{eq:RecurrenceBessel}. With a bit of rearranging and the 
    integral substitution $s=r/R$, we conclude that 
    \begin{equation*}
        \sqrt{s}h_M(s) = \sum_{m=1}^{M}\frac{2}{J_{n+1/2}^2(\lambda_{n,m})}
        c_mJ_{n+1/2}(\lambda_{n,m}s),
    \end{equation*}
    where $c_m$ is given by
    \begin{equation*}
        c_m := \int_{0}^{1}s\left(\sqrt{s}h(s)\right)J_{n+1/2}(\lambda_{n,m}s)\di s
        \quad\text{for every }m\in\N.
    \end{equation*}
    If we now define $g(s):=\sqrt{s}h(s)$, then this function vanishes in a 
    neighborhood of $s=0$, since $\textup{supp}(g) = \textup{supp}(h)\subsetneq (0,1)$. 
    We conclude that $g\in\Cont^1[0,1]$, so we can apply 
    Theorem \ref*{thm:SeriesContBessel}, which proves that 
    $\sqrt{s}h_M(s)\to\sqrt{s}h(s)$ as $M\to\infty$ for every $s\in(0,1)$. 
    It is then obvious that $H_M(r)\to H(r)$ as $M\to\infty$ for every $r\in(0,R)$. 

    In the second part of the proof we will show that the sequence $H_M$ is convergent in 
    $\LSq_w(0,R)$. We can quantify the decay of the Fourier coefficients using 
    Theorem \ref*{thm:DecayFourierCoeffBessel}, because $h\in\Cont^1[0,1]$ and
    \begin{align*}
        \langle H,B_{m,n}\rangle_{\LSq_w(0,R)} &= 
        R^3\int_{0}^{1}s^2H(Rs)B_{m,n}(Rs)\di s \\
        &= R^{3/2}\sqrt{\frac{2}{j_n^2(\lambda_{n,m})}}
        \int_{0}^{1}s^2h(s)j_n(\lambda_{n,m}s)\di s,
    \end{align*}
    where we used the substitution $s=r/R$ again. Therefore a natural number $K$ and a 
    constant $C>0$ exist such that 
    \begin{align*}
        \left|\langle H,B_{m,n}\rangle_{\LSq_w(0,R)}\right|\leq\frac{C}{m-3/4}\quad
        \text{for every }m\geq K.
    \end{align*}
    Now, for two natural numbers $N>M\geq K$, we use the orthonormality of the $B_{m,n}$ 
    to conclude that 
    \begin{equation*}
        \left\|H_N-H_M\right\|^2_{\LSq_w(0,R)}
        \leq \sum_{m=M+1}^{N}\frac{C^2}{(m-3/4)^2}\limto[{N,M}]0.
    \end{equation*}
    We have thus proven that $(H_M)_{M\in\N}$ is a Cauchy sequence in $\LSq_w(0,R)$, 
    meaning the limit 
    \begin{align*}
        \sum_{m=1}^{\infty}\langle H,B_{m,n}\rangle_{\LSq_w(0,R)}B_{m,n}
    \end{align*}
    exists as a $\LSq_w(0,R)$-function. It is well known that there will be a subsequence 
    which converges pointwise almost everywhere to this $\LSq_w(0,R)$-limit (see for example 
    the proof of Theorem 1.3 in Chapter I on pages 5--6 in \cite{SteinFuncAna}). However this 
    subsequence will also converge to $H$ pointwise, as we have previously seen, meaning 
    $H$ has to be the $\LSq_w(0,R)$-limit of $H_M$.
    Since $\Cont^2_{\textup{c}}(0,R)\subs \LSq_w(0,R)$ is dense (this follows from the fact 
    that $\Cont^{\infty}_{\textup{c}}(0,R)$ is dense in $\LSq_w(0,R)$, which is stated in Corollary
    3 to Theorem 15.3 on page 159 in \cite{TrevesTopVectorSp}), an 
    $\varepsilon/3$-argument shows that the statement is true for a general 
    $H\in\LSq_w(0,R)$. 
\end{proof}
\begin{theoremS}\label{thm:ONBEF}
    The functions 
    \begin{equation}\label{eq:ONBEF}
        u_{m,n,j}(\bx) = u_{m,n,j}(r\bi) := B_{m,n}(r)Y_{n,j}(\bi)
    \end{equation}
    for $\bx\neq 0$ and $m\in\N$, $n\in\N_0$ and $j\in\{-n,\dots,n\}$ are an orthonormal 
    basis of $\LSq(B)$.
\end{theoremS}
\begin{proof}[Proof]
    The $Y_{n,j}$ form a basis of $\LSq(\S^2)$, the $B_{m,n}$ a basis of 
    $\LSq_w(0,R)$ and the inner product of $\LSq(B)$ can be decomposed for functions 
    $u(r\bi)=F(r)Y(\bi)$ and $v(r\bi)=G(r)Z(\bi)$ in the following way:
    \begin{equation*}
        \langle u,v\rangle_{\LSq(B)} = \langle F,G\rangle_{\LSq_w(0,R)}
        \langle Y,Z\rangle_{\LSq(\S^2)}.
    \end{equation*}
    It is then easy to see that the functions $u_{m,n,j}$ are indeed orthonormal in 
    $\LSq(B)$. For completeness, the implication 
    \begin{center}
        $\langle V,u_{m,n,j}\rangle_{\LSq(B)}=0$ 
        for every $m\in\N$, $n\in\N_0$ and $j\in\{-n,\dots,n\}$ $\Rightarrow$ $V=0$ in 
        $\LSq(B)$
    \end{center}
    is proven by exploiting the basis properties of both $B_{m,n}$ and $Y_{n,j}$.
\end{proof}
\begin{lemmaS}
    For any gravitational potential $V$ defined as in Equation \eqref{eq:GravPot} 
    we have 
    \begin{equation}\label{eq:LaplaceSymmPot}
        \langle \Delta V,u_{m,n,j}\rangle_{\LSq(B)} = 
        -\gamma^2_{n,m}\langle V,u_{m,n,j}\rangle_{\LSq(B)}
    \end{equation}
    for every $m\in\N$, $n\in\N_0$ and $j\in\{-n,\dots,n\}$.
\end{lemmaS}
\begin{proof}[Proof]
    If $V$ is defined as in Equation \eqref{eq:GravPot}, then it must satisfy 
    the properties I.--III. listed in Chapter \ref{sec:GravPot}.We also need the eigenvalue 
    property of our basis functions, namely that 
    $\Delta u_{m,n,j}(\bx) = -\gamma_{n,m}^2u_{m,n,j}(\bx)$ holds for every 
    $\bx\in B\setminus\{0\}$. We will then have to show that 
    the term
    \begin{equation*}
        \langle \Delta V,u_{m,n,j}\rangle_{\LSq(B)}-\langle V,\Delta u_{m,n,j}
        \rangle_{\LSq(B)}
    \end{equation*}
    vanishes for every valid index. To prove this, Green's second identity 
    seems useful, but our basis functions are in general not continuous at 
    $\bx = 0$, so we have to exclude a small neighborhood of this point 
    and use a limiting process. For this we define $D_{\eps}:=B\setminus B_{\eps}(0)$
    for every $0<\eps<R$ and study the term
    \begin{equation*}
        L_{m,n,j}(\eps) := \int_{D_{\eps}}\Delta V(\bx)u_{m,n,j}(\bx)-
        V(\bx)\Delta u_{m,n,j}(\bx)\di\bx,
    \end{equation*}
    where $\bm{n}$ is the outer unit normal to $\partial D_{\eps}$.
    Now since $V,u_{m,n,j}\in\Cont^2(D_{\eps})\cap\Cont^1(\overline{D_{\eps}})$
    is satisfied, we can use Green's second identity and conclude that 
    \begin{equation*}
        L_{m,n,j}(\eps) = \int_{\partial D_{\eps}}u_{m,n,j}(\bx)\frac{\partial V}
        {\partial\bm{n}}(\bx)-V(\bx)\frac{\partial u_{m,n,j}}{\partial\bm{n}}
        (\bx)\di\textup{S}(\bx)=: R_{m,n,j}(\eps).
    \end{equation*}
    Because all of the occuring integrands are bounded on $B$, we then estimate
    \begin{align*}
        \left|\int_B \Delta V(\bx)u_{m,n,j}(\bx)-V(\bx)\Delta u_{m,n,j}(\bx)\di\textup{S}(\bx)
        -L_{m,n,j}(\eps)\right| &\leq K\eps^3 \\
        \left|\int_{\partial B}u_{m,n,j}(\bx)\frac{\partial V}
        {\partial\bm{n}}(\bx)-V(\bx)\frac{\partial u_{m,n,j}}{\partial\bm{n}}
        (\bx)\di\textup{S}(\bx)-R_{m,n,j}(\eps)\right| &\leq D\eps^2
    \end{align*}
    for some constants $K,D>0$. Letting $\eps\to 0$ we see that 
    \begin{align*}
        \langle \Delta V,u_{m,n,j}\rangle_{\LSq(B)}-\langle V,\Delta u_{m,n,j}
        \rangle_{\LSq(B)} = 
        \int_{\partial B}u_{m,n,j}(\bx)\frac{\partial V}
        {\partial\bm{n}}(\bx)-V(\bx)\frac{\partial u_{m,n,j}}{\partial\bm{n}}
        (\bx)\di\textup{S}(\bx).
    \end{align*}
    To show that the boundary integral vanishes, we substitute
    $\frac{\partial}{\partial \bm{n}}=\partial_r$, $\bx=R\bi$ and 
    $\di\textup{S}(\bx) = R^2\di\omega(\bi)$ as well as 
    $u_{m,n,j}(r\bi) = B_{m,n}(r)Y_{n,j}(\bi)$. Then 
    \begin{align*}
        & \langle \Delta V,u_{m,n,j}\rangle_{\LSq(B)}-\langle V,\Delta u_{m,n,j}
        \rangle_{\LSq(B)}\\
        = & R^2B_{m,n}(R)\langle Y_{n,j},\partial_rV(r\cdot)\rangle_{\LSq(\S^2)}\big|_{r=R}
        -R^2B_{m,n}'(R)\langle Y_{n,j},V(R\cdot)\rangle_{\LSq(\S^2)}.
    \end{align*}
    Finally, because of Equation \eqref{eq:BoundCondPot} we have 
    \begin{equation*}
        \langle Y_{n,j},\partial_rV(r\cdot)\rangle_{\LSq(\S^2)}\big|_{r=R}
        = -\frac{n+1}{R}\langle Y_{n,j},V(R\cdot)\rangle_{\LSq(\S^2)}
    \end{equation*}
    and because of Equation \eqref{eq:BC} the identity $B_{m,n}'(R) = 
    -\frac{n+1}{R}B_{m,n}(R)$ holds. This proves that the boundary integral indeed 
    vanishes and thus concludes the proof.
\end{proof}

\subsection{Determining the Wind-Induced Potential}\label{sec:DetDynPot}
In this chapter we will deal with the problem of determining $V'$ from the thermo-gravitational
source term $S^u$ in the case of using the Thermo-Gravitational Wind Equation as the model for 
the dynamic problem.
\begin{theoremS}
    The Helmholtz 
    equation
    \begin{equation*}
        \left(\Delta+\left(\frac{\pi}{R}\right)^2\right)V' = S^u+\eta\quad
        \text{on }B
    \end{equation*}
    from Chapter \ref{sec:TGWEq} is equivalent to the set of equations
    \begin{equation}\label{eq:HelmholtzFreq}
        \left(\left(\frac{\pi}{R}\right)^2-\gamma_{n,m}^2\right)
        \langle V',u_{m,n,j}\rangle_{\LSq(B)} = \langle S^u,u_{m,n,j}\rangle_{\LSq(B)}
        +\eta_m\delta_{n,0}\delta_{j,0}
    \end{equation}
    for every $m\in\N$, $n\in\N_0$ and $j\in\{-n,\dots,n\}$. The integration constants $\eta_m\in\R$ are given by $\eta_m := \sqrt{4\pi}\langle\eta,G_{m,0}\rangle_{\LSq_w(0,R)}$ and are therefore undetermined. 
\end{theoremS}
\begin{proof}[Proof]
    We just project the Helmholtz equation onto our basis functions $u_{m,n,j}$ and use 
    Equation \eqref{eq:LaplaceSymmPot}. This yields 
    \begin{equation*}
        \left(\left(\frac{\pi}{R}\right)^2-\gamma_{n,m}^2\right)
        \langle V',u_{m,n,j}\rangle_{\LSq(B)} = \langle S^u,u_{m,n,j}\rangle_{\LSq(B)}
        +\langle \eta,u_{m,n,j}\rangle_{\LSq(B)}    
    \end{equation*}
    for every $m\in\N$, $n\in\N_0$ and $j\in\{-n,\dots,n\}$. Since $\eta$ is a function 
    only of the radius $r$, the basis coefficients of this function vanish for $n\neq 0$
    and $j\neq 0$. A simple calculation shows that 
    $\langle\eta,u_{m,0,0}\rangle_{\LSq(B)}=
    \sqrt{4\pi}\langle\eta,B_{m,0}\rangle_{\LSq_w(0,R)}$.
\end{proof}
Since we assume that $V'$ does not depend on $\varphi$ in spherical coordinates, Equation 
\eqref{eq:HelmholtzFreq} can be restated as
\begin{equation*}
    \left(\left(\frac{\pi}{R}\right)^2-\gamma_{n,m}^2\right)
    \langle V',u_{m,n,0}\rangle_{\LSq(B)} = \langle S^u,u_{m,n,0}\rangle_{\LSq(B)}
    +\eta_m\delta_{n,0}
\end{equation*}
for every $m\in\N$ and $n\in\N_0$. We would like to solve this equation for the coefficients of 
the perturbed potential $V'$, so we need to guarantee that $\gamma_{n,m}^2\neq 
\left(\frac{\pi}{R}\right)^2$ for every value of $n$ and $m$. However the coefficients $\gamma_{n,m}$ solve the equation $j_{n-1}(\gamma_{n,m}R)=0$, so $\gamma_{1,1}=\pi/R$.
This implies that we need another assumption to deal with the coefficient $\langle V',u_{1,1,0}\rangle_{\LSq(B)}$.
\begin{assumptionS}\label{asmp:PolytropeEigenvalue}
    In the following we assume that 
    \begin{equation*}
        \langle V',u_{1,1,0}\rangle_{\LSq(B)} = 0
    \end{equation*}
    holds (a physical justification of this assumption is given in \cite*{ZonalWindsWicht}).
\end{assumptionS}
\begin{remarkS}\label{rem:UniquePot}
    Uniquely determining the potential $V'$ is not always necessary, but if we wanted to 
    do that, we would need another assumption to eliminate the unknown coefficients 
    $\eta_m$. The condition $\langle V',u_{m,0,0}\rangle_{\LSq(B)} = 0$ for every 
    $m\in\N$ or equivalently stated $\int_{\S^2}V'(r\bi)\di\omega(\bi) = 0$ for every 
    $r\in(0,R)$ is sufficient for this. Under these assumptions and remembering the 
    $\varphi$-independece of $V'$ the uniquely determined dynamic potential has the form
    \begin{equation}\label{eq:UniquePot}
        V' = \sum_{\substack{m,n\in\N \\ (m,n)\neq (1,1)}}
        \frac{\langle S^u,u_{m,n,0}\rangle_{\LSq(B)}}{\left(\frac{\pi}{R}\right)^2-
        \gamma_{n,m}^2}u_{m,n,0}\quad\text{in }\LSq(B).
    \end{equation}
\end{remarkS}
It is then also relatively easy to see that 
\begin{equation*}
    u_{m,n,0}(r\bi)= \sqrt{\frac{2n+1}{2\pi R^3}}\frac{j_n(\gamma_{n,m}r)}
    {|j_n(\gamma_{n,m}R)|}P_n(\xi_3)
\end{equation*}
holds for every $n\in\N_0$ and $m\in\N$.
\newpage

\section{Connecting the \texorpdfstring{$J_n$}{Jn} to the Wind Field}
\label{sec:JnConnect}
In this chapter, we will see how to formulate the relationship between the wind field and 
the gravitational data as an inverse problem for some parameters of the wind field.

\subsection{A Closer Look at the \texorpdfstring{$J_n$}{Jn}}
\label{sec:JNOptim}
We will first recapitulate the definition of the gravitational coefficients
$J_n$, which first occurred in Equation \eqref{eq:PotExpan}.
\begin{remarkS}\label{rem:GravCoeff}
    For a mass density $\rho\in\Cont^2(\overline{B})$ satisfying Assumption
    \ref*{asmp:PhiIndependence} the (tesseral) gravitational 
    coefficients caused by $\rho$ are 
    \begin{align*}
        J_n = -\frac{1}{MR^n}\int_{B}|\bx|^nP_n\left(\frac{x_3}{|\bx|}\right)
        \rho(\bx)\di\bx
    \end{align*}
    for every $n\in\N$ with $n\geq 2$.
    These coefficients can be derived from the measured gravitational data of an orbiter like Juno and are therefore available for further processing.
\end{remarkS}
Since these coefficients first occurred, we have split up the density $\rho$ into 
two terms $\rho_0$ and $\rho'$, corresponding to a perturbation approach. Because of linearity,
we can split up the $J_n$ in the same way.
\begin{definitionS}\label{def:StatDynGravCoeffs}
    For the mass density $\rho_0\in\Cont^2(\overline{B})$ from
    Definition \ref*{def:StatQuant} we define the 
    \textbf{hydrostatic background gravitational coefficients} as 
    \begin{align*}
        J^0_n:=-\frac{1}{MR^n}\int_{B}|\bx|^nP_n\left(\frac{x_3}{|\bx|}\right)
        \rho_0(\bx)\di\bx.
    \end{align*}
    For the dynamic or wind-induced density $\rho'\in\Cont^2(\overline{B})$ from
    Definition \ref*{def:DynQuant} we define the 
    \textbf{dynamic gravitational coefficients} as
    \begin{align*}
        \delta J_n:=-\frac{1}{MR^n}\int_{B}|\bx|^nP_n\left(\frac{x_3}{|\bx|}
        \right)\rho'(\bx)\di\bx.
    \end{align*}
    Both sets of coefficients are defined for $n\in\N$ with $n\geq 2$.
\end{definitionS} 
\begin{remarkS}
    If $G'(\bx) = \rho'(\bx)+\eta(|\bx|)$ represents a radial indeterminacy of the 
    density, then the coefficients $\delta J_n$ remain unchanged, because 
    \begin{align*}
        & \int_B|\bx|^n\eta(|\bx|)P_n\left(\frac{x_3}{|\bx|}\right)\di\bx \\
        = & \int_{0}^{R}r^{n+2}\eta(r)\di r\int_{\S^2}P_n(\xi_3)\di\omega(\bi) \\
        = & 2\pi\int_{0}^{R}r^{n+2}\eta(r)\di r\int_{-1}^{1}P_n(t)\di t = 0
    \end{align*}
    for every $n\geq 2$ by orthogonality of the $P_n$ in $\LSq[-1,1]$ ($P_0(t)\equiv 1$).
\end{remarkS}
It is worth noting that the $\delta J_n$ are the gravity moments caused by $\rho'$. If 
we use the Thermal Wind Equation as the basis of our dynamic model, then this presents
no problem since we can determine $\rho'$ in terms of $\bm{u}$ by Equation 
\eqref{eq:PoissonKaspi}. However if we use the Thermo-Gravitational Wind Equation we only 
determine $V'$ in dependence of $\bm{u}$, not $\rho'$. Thus we need a mathematical 
connection between $V'$ and the $\delta J_n$. We present a relevant method of calculation now
by first recapitulating the properties of $V'$ in the following remark.
\begin{remarkS}\label{rem:ProperiesWindPot}
    The wind-induced gravitational potential $V'$ satisfies 
    $V'\in\Cont^2(B)\cap\Cont^2(\R^3\setm\overline{B})\cap\Cont^1(\R^3)$, is regular at
    infinity and 
    \begin{align*}
        \Delta V'(\bx) = 0\quad\text{for every }\bx\in\R^3\setm\overline{B}.
    \end{align*}
    On $B$, the potential satisfies
    \begin{align*}
        \Delta V' = 4\pi G\rho',
    \end{align*}
    where $\rho'\in\Cont^2(\overline{B})$ is the wind-induced density.
\end{remarkS}
\begin{lemmaS}\label{lma:GravCoeffPot}
    If $V'$ satisfies the properties in Remark \ref*{rem:ProperiesWindPot}, then
    we can express the coefficients $\delta J_n$ via 
    \begin{align*}
        \delta J_n = \frac{2n+1}{4\pi}\,\frac{R}{GM}\int_{\S^2}V'(R\bi)P_n(\xi_3)
        \di\omega(\bi)
    \end{align*}
    for every $n\geq 2$.
\end{lemmaS}
\begin{proof}[Proof]
    For $n\geq 2$ we define the functions $H_n^{\intTop}(\bx):=|\bx|^nP_n(x_3/|\bx|)$ for 
    $\bx\neq 0$ and $H_n^{\intTop}(0):=0$, these functions are harmonic on $B$. Then if we 
    use the fact that $\Delta V' = 4\pi G\rho'$ on $B$ we get
    \begin{align*}
        -4\pi GMR^n\delta J_n &= \int_BH_n^{\intTop}(\bx)\Delta V'(\bx)\di\bx
        -\int_BV'(\bx)\Delta H_n^{\intTop}(\bx)\di\bx \\
        &= \int_{\partial B}H_n^{\intTop}(\bx)\frac{\partial V'}{\partial\bm{n}}(\bx)
        -V'(\bx)\frac{\partial H_n^{\intTop}}{\partial\bm{n}}(\bx)\di\textup{S}(\bx).
    \end{align*}
    We used Green's second identity, since both $H_n^{\intTop}$ and $V'$ are contained in 
    $\Cont^2(B)\cap\Cont^1(\overline{B})$, in the latter equation. Since 
    $\frac{\partial}{\partial\bm{n}} = \partial_r$ and $Y_{n,0}(\bi)=
    \sqrt{\frac{2n+1}{4\pi}}P_n(\xi_3)$, a bit of calculation shows that the boundary 
    integrals are equal to
    \begin{equation*}
        \sqrt{\frac{4\pi}{2n+1}}\left(R^{n+2}\langle \partial_rV'(r\,\cdot),
        Y_{n,0}\rangle_{\LSq(\S^2)}\big|_{r=R}-nR^{n+1}\langle V'(R\,\cdot),Y_{n,0}
        \rangle_{\LSq(\S^2)}\right).
    \end{equation*}
    Since $V'$ satisfies Equation \eqref{eq:BoundCondPot}, we know that 
    \begin{equation*}
        \langle \partial_rV'(r\,\cdot),Y_{n,0}\rangle_{\LSq(\S^2)}\big|_{r=R} = 
        -\frac{n+1}{R}\langle V'(R\,\cdot),Y_{n,0}\rangle_{\LSq(\S^2)}.
    \end{equation*}
    This allows us to conclude that
    \begin{equation*}
        \delta J_n = \frac{2n+1}{4\pi}\,\frac{R}{GM}\int_{\S^2}P_n(\xi_3)V'(R\bi)
        \di\omega(\bi).\qedhere
    \end{equation*}
\end{proof}
Having proven this lemma, we can relate the $\delta J_n$ to the coefficients
of $V'$ in the basis $u_{m,n,j}$ we constructed in Section \ref*{sec:ONBGravPot}.
\begin{lemmaS}\label{lma:DeltaJNExpansion}
    If $V'$ satisfies the properties in Remark \ref*{rem:ProperiesWindPot}, then we can 
    calculate the $\delta J_n$ via 
    \begin{equation*}
        \delta J_n = \sqrt{\frac{2n+1}{4\pi}}\,\frac{R}{GM}\sum_{m=1}^{\infty}
        \langle V',u_{m,n,0}\rangle_{\LSq(B)}B_{m,n}(R)
    \end{equation*}
    for every $n\geq 2$.
\end{lemmaS}
\begin{proof}[Proof]
    For $n\in\N_0$ with $n\geq 2$ we define the function $f_n$ by 
    \begin{align*}
        f_n:[0,R]\to \R,\,r\mapsto f_n(r):=\int_{\S^2}V'(r\bi)P_n(\xi_3)
        \di\omega(\bi),
    \end{align*}
    then since $V'\in\Cont^1(\overline{B})$, this function is in $\Cont^1[0,R]\subs 
    \LSq_w(0,R)$. By Lemma \ref*{lma:GravCoeffPot} we have $\delta J_n = 
    \frac{2n+1}{4\pi}\,\frac{R}{GM}f_n(R)$. The $M$-th partial Fourier sum of $f_n$ in 
    terms of the basis $(B_{m,n})_{m\in\N}$ is 
    \begin{equation*}
        F_{M,n}(r) := \sum_{m=1}^{M}\langle f_n,B_{m,n}\rangle_{\LSq_w(0,R)}B_{m,n}(r).
    \end{equation*}
    By exactly the same transformations as in the proof of Theorem \ref*{thm:ONBRadial}, 
    we see that 
    \begin{align*}
        \sqrt{s}F_{M,n}(Rs) = \sum_{m=1}^{M}\frac{2}{J_{n+1/2}^2(\lambda_{n,m})}c_{m,n} 
        J_{n+1/2}(\lambda_{n,m}s),
    \end{align*}
    where $s:=r/R$, $\lambda_{n,m}:=\gamma_{n,m}R$ and 
    \begin{align*}
        c_{n,m} := \int_{0}^{1}s(\sqrt{s}f_n(Rs))J_{n+1/2}(\lambda_{n,m}s)\di s
    \end{align*}
    for every $m\in\N$. This sum has the form of a Dini-Bessel series as in the proof of 
    Theorem \ref*{thm:SeriesContBessel}. This means that 18.34 in \cite*{BesselWatson} 
    allows the conclusion that $F_{M,n}(R)\to f_n(R)$ for $M\to\infty$, since the 
    function $s\mapsto \sqrt{s}f_n(Rs)$ is continuously differentiable and thus of 
    bounded variation in intervals of the form $[\varepsilon,1]$ for every $\varepsilon$ 
    between $0$ and $1$. Using $Y_{n,0}(\bi)=\sqrt{\frac{2n+1}{4\pi}}P_n(\xi_3)$ we see
    that $\langle f_n,B_{m,n}\rangle_{\LSq_w(0,R)} = \sqrt{\frac{4\pi}{2n+1}}\langle V',
    u_{m,n,0}\rangle_{\LSq(B)}$. This implies that
    \begin{equation*}
        \delta J_n = \frac{2n+1}{4\pi}\,\frac{R}{GM}f_n(R) = 
        \sqrt{\frac{2n+1}{4\pi}}\,\frac{R}{GM}\sum_{m=1}^{\infty}\langle V',u_{m,n,0}
        \rangle_{\LSq(B)}B_{m,n}(R).\qedhere
    \end{equation*}
\end{proof}
\begin{theoremS}\label{thm:DynJnCalc}
    If $V'$ satisfies the properties in Remark \ref*{rem:ProperiesWindPot}, then the 
    $\delta J_n$ have the form
    \begin{equation*}
        \delta J_n = \sqrt{\frac{2n+1}{4\pi}}\,\frac{R}{GM}\sum_{m=1}^{\infty}
        \frac{\langle S^u,u_{m,n,0}\rangle_{\LSq(B)}}{\left(\frac{\pi}{R}\right)^2-
        \gamma_{n,m}^2}B_{m,n}(R)
    \end{equation*}
    for every $n\geq 2$.
\end{theoremS}
\begin{proof}[Proof]
    We just apply the results of Section \ref{sec:DetDynPot} to 
    Lemma \ref{lma:DeltaJNExpansion}. Note that we just need Assumption 
    \ref{asmp:PolytropeEigenvalue}, the last summand from Equation 
    \eqref{eq:HelmholtzFreq} (the one containing $\eta_m$) vanishes because $n\geq 2$.   
\end{proof}

\subsection{An Inverse Problem}\label{sec:InvProb}
The gravity harmonics provided by spacecraft measurements are not sufficient to constrain the depth structure of the zonal winds. 
The inversion would indeed be highly non-unique and additional information and assumptions on the flow are required.
Numerical simulations and theoretical consideration show that the flow tends to be geostrophic. This means that there is practically no variation in the direction of the rotation axis, $\bm{e}^3$. 
This is a result of the strong Coriolis force, resulting in a primary balance between the Coriolis force and the non-hydrostrophic pressure gradient. 
We first specify the structure of the zonal wind field $u_{\varphi}$ according to these assumptions, this structure is also assumed in \cite*{CylindWindsKaspi,AtmosWindsKaspi} but presents some problems from a mathematical point of view (see Remark \ref*{rem:ProblemCylindWind}).
These models use the surface wind field as a function of $t=\cos(\vartheta)$ (in spherical coordinates) $u^{\surf}:[-1,1]\to\R,\quad t\mapsto u^{\surf}(t)$, which is known from direct observation (for Saturn and Jupiter see \cite*{SurfaceWindSaturn} and \cite*{SurfaceWindJupiter} respectively). 
In the following definition we specify the wind structure that is also assumed in \cite{CylindWindsKaspi,AtmosWindsKaspi}.
\begin{figure}[!ht]
    \centering
    \begin{subfigure}[ht]{0.45\textwidth}
        \centering
        \includegraphics[width=\textwidth]{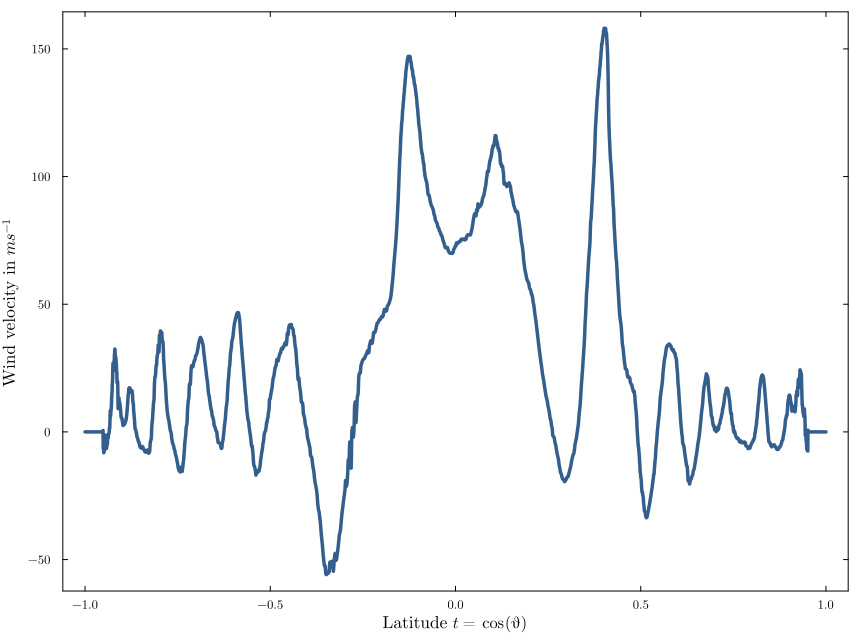}
        \caption*{Surface wind on Jupiter (data from \cite{SurfaceWindJupiter})}
    \end{subfigure}
    \hfill
    \begin{subfigure}[ht]{0.45\textwidth}
        \centering
        \includegraphics[width=\textwidth]{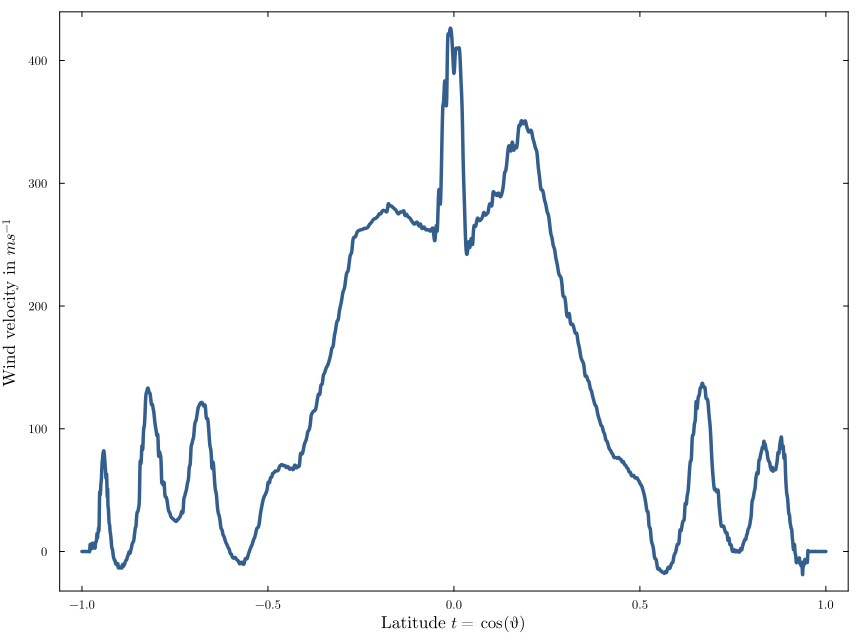}
        \caption*{Surface wind on Saturn (data from \cite{SurfaceWindSaturn})}
    \end{subfigure}
    \caption{The surface wind on Jupiter and Saturn as 
    functions of the latitude $t=\cos(\vartheta)$}\label{fig:SurfacWinds}
\end{figure}
\begin{definitionS}\label{def:ProjWindCylind}
    We define the \textbf{cylindrically projected wind} as the function
    \begin{equation*}
        u^{\proj}_{\textup{c}}:B_R(0)\setm \{0\}\to\R,\quad r\bi\mapsto u^{\surf}
        \left(\sgn(\xi_3)\sqrt{1-\frac{r^2}{R^2}(1-\xi_3^2)}\right).
    \end{equation*}
    See also Figure \ref{fig:UProjExplained}.
\end{definitionS}
\begin{figure}[!ht]
    \centering
    \includegraphics[width=0.5\textwidth]{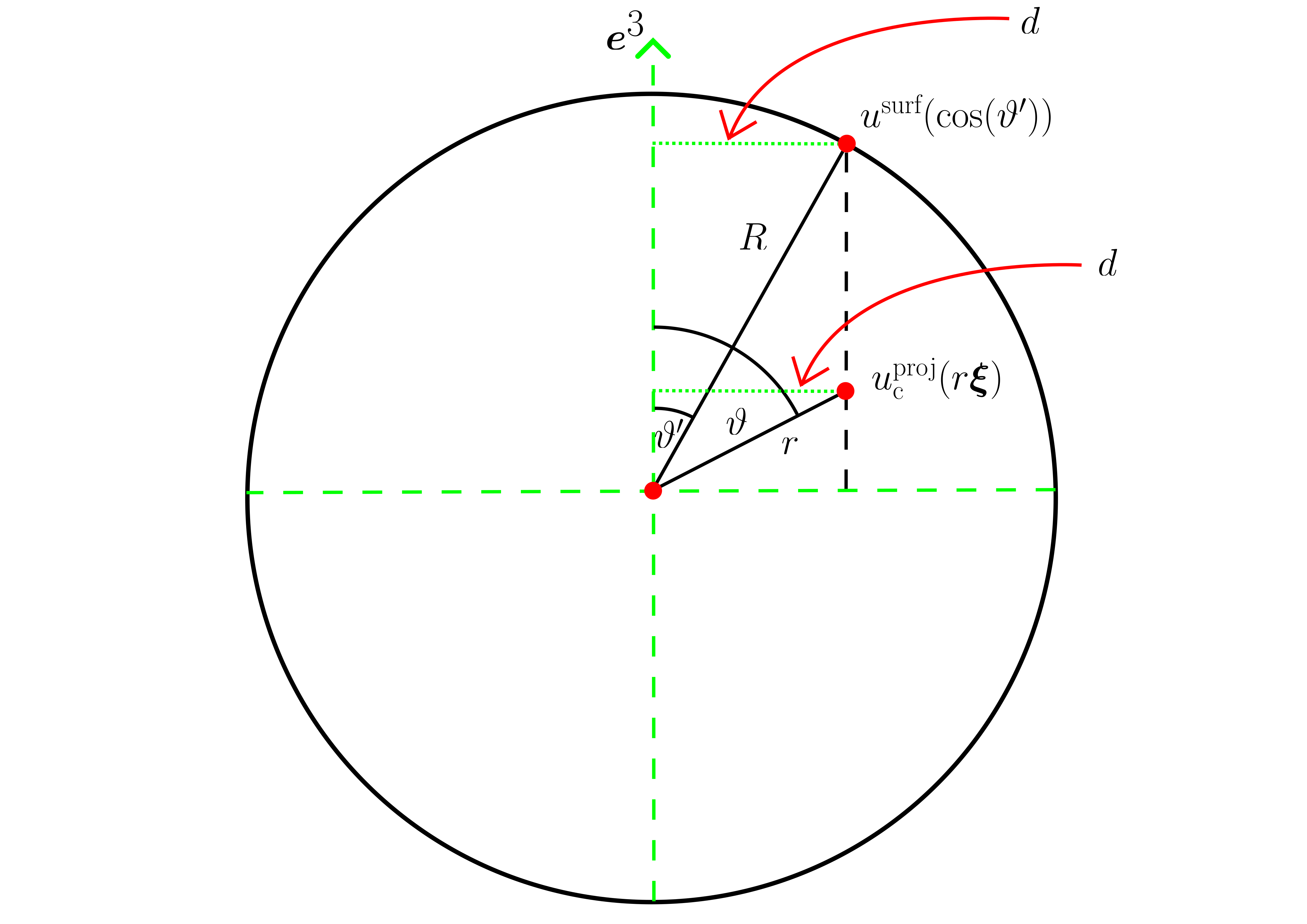}
    \caption{How to calculate $u^{\proj}_{\textup{c}}$ from $u^{\surf}$}
    \label{fig:UProjExplained}
\end{figure}
The value of $u^{\surf}$ that defines $u^{\proj}_c(r\bi)$ is the value that the surface 
wind takes on the point closest to $r\bi$ on the planetary surface that is also on the 
line parallel to the axis of rotation of the planet and through $r\bi$.
In Figure \ref{fig:UProjExplained} this means that $u^{\proj}_c(r\bi(\vartheta,\varphi))$ is calculated via $u^{\proj}_c(r\bi(\vartheta,\varphi)) = u^{\surf}(\cos(\vartheta'))$, where $u^{\surf}$ is the measured surface wind from above.
The angle $\vartheta'$ is defined by the equation $d = r\sin(\vartheta) = R\sin(\vartheta')$, meaning that $\cos(\vartheta') = \sgn(\xi_3)\sqrt{1-\frac{r^2}{R^2}(1-\xi_3^2)}$.
\begin{remarkS}\label{rem:ProblemCylindWind}
    The definition of the cylindrically projected wind presents a mathematical problem,
    namely a discontinuity at the equatorial plane. This arises because the surface wind 
    field is not symmetrical with respect to the equator. Thus projecting the wind field 
    inward parallel to the axis of rotation leads to different values of 
    $u^{\proj}_{\textup{c}}$ when approaching the equatorial plane from the North and the 
    South. This fact contradicts our assumption that $\bm{u}\in 
    \Cont^2(\overline{B},\R^3)$, which we made at the beginning of Chapter 
    \ref*{sec:WindInsidePlanet}. Since the formulas we derived for calculating 
    $\rho'$ (or equivalently $V'$) all include a term which includes a spatial derivative of
    $\bm{u}$ or $\beps^{\varphi}\cdot\bm{u}=u_{\varphi}$, see Equation 
    \eqref{eq:TWSimplestSolStep1} and Definition \ref*{def:SourceTGW}, this is 
    especially problematic. These facts have also been discussed in the literature 
    \cite*{ZonalWindsWicht2,OddCoeffsDiscontKong} and the discontinuity can contribute a significant 
    amount to the calculation of the coefficients $J_n$ even when some smoothing is 
    applied along the discontinuity (see also \cite*{OddCoeffsDiscontKong}).
\end{remarkS}
The projected wind fields are then usually multiplied by a function $Q_p$, which 
characterizes the decrease of the wind field in the interior of the planet 
(see \cite*{CylindWindsKaspi,AtmosWindsKaspi}). These functions are dependent
on parameters $p\in D\subs\R^d$, which can be used to fit the wind to the measured 
gravitational coefficients. So we are now assuming that $u_{\varphi}$ has the form
\begin{equation*}
    u_{\varphi}(\bx) = Q_p(\bx)u^{\proj}_c(\bx).
\end{equation*}
We can then calculate the dynamical gravitational coefficients in dependence of 
$p\in D$ using Theorem \ref{thm:DynJnCalc}, denoting these coefficients as $\delta J_n(p)$ 
for every $n\geq 2$ and $p\in D$. 
Since the gravitational coefficients $J_n$ of the gas giants Jupiter and Saturn can be calculated from the measurements done by the Juno and the Cassini missions respectively and are known up to $J_{10}$, see \cite*{JupiterGravFieldIess,SaturnGravFieldIess,AtmosWindsKaspi} (more recently these coefficients have been determined up to $J_{40}$, see \cite*{CylindWindsKaspi}), we can compare the calculated coefficients to the measured ones and see which parameter $p$ gives the best fit to the data. 
Now we cut off the sequences of calculated coefficients $(J^0_n)_{n\geq 2}$ and $(\delta J_n(p))_{n\geq 2}$ at some index $N\in\N$, denoting the resulting vectors in $\R^{N-2}$ as $J^0$ and $\delta J(p)$ respectively. 
We can then compare these to the vector of measured coefficients $J:=(J_n)_{n=2,\ldots,N}\in\R^{N-2}$ using some functional, for example
\begin{equation*}
    \mathcal{L}(p) := \|J-(J^0+\delta J(p))\|_2^2.
\end{equation*}
Minimizing this functional will then yield the parameter $p_0$, defining the wind field, which best fits the data.
The calculation scheme for this inverse problem is summarized in Figure \ref{fig:FlowchartIVP}.
\begin{figure}[ht]
    \centering
    \includegraphics[width=0.7\textwidth]{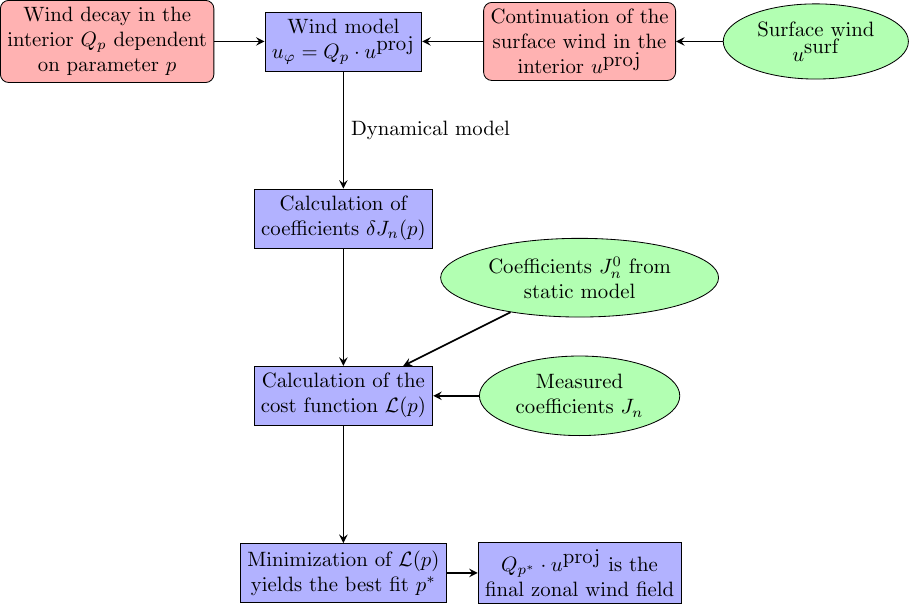}
    \caption{Flowchart of the inverse problem}\label{fig:FlowchartIVP}
\end{figure}
\newpage

\section{Conclusion and Further Thoughts}
This article provides a mathematically rigorous framework for the development of solutions of a relatively simple model (the one developed in Section \ref*{sec:TGWEq}), which describes the relationship between the gravitational potential and the wind field in the atmosphere of a gas giant, like Jupiter or Saturn. 
In particular, we constructed an orthonormal basis that is specifically adapted to solve this problem (in Section \ref{sec:ONBGravPot}) and provided a method to calculate the gravitational coefficients generated by the wind using this basis (see Section \ref{sec:DetDynPot}).
\par
As one can see from a glance at Figure \ref{fig:FlowchartFWP}, a lot of assumptions are necessary to reach a concrete mathematical connection between the gravitational coefficients and the zonal winds.
What follows now is a short discussion of how some of these assumptions could be relaxed in further research.
\par 
We assumed a polytrope index unity background model. However, \cite{ZonalWindsWicht} show that the method discussed would also work, as long as the normalized density gradient,
\begin{equation*}
    (\partial_r \rho_0)/g_0,
\end{equation*}
is roughly independent of the radius.
\par
Another assumption which could be changed is the parametrization of the wind field in the interior of the planet.
We can of course observe the wind on the surface $u^{\surf}$ and in this paper we chose a very specific way of continuing the wind field into the interior, namely along cylinders with axes collinear to the axis of rotation of the planet.
This yielded the quantity $u^{\proj}_{\textup{c}}$ described in Definition \ref{def:ProjWindCylind}.
On closer inspection this leads to a mathematical problem because the surface wind fields on Jupiter and Saturn are not symmetric upon reflection on the equatorial plane, see Figure \ref{fig:SurfacWinds}.
Our method of constructing $u^{\proj}_{\textup{c}}$ means that just north of the equatorial plane the values of the wind field will be derived from the surface winds on the northern hemisphere, but just south the values correspond to the surface winds on the southern hemisphere.
The asymmetry in the surface wind field thus leads to a discontinuity of $u^{\proj}_c$ along the equatorial plane. If our decay function $Q_p$ does not compensate this step, the zonal winds $u_{\varphi}$ will also jump across this plane.
This is problematic because our coefficients $\delta J_n$ depend on the source term $S^u$ via Theorem \ref{thm:DynJnCalc} and the source term $S^u$ is defined by an integral over $\nabla\cdot(\rho_0 u_{\varphi})$ (see Definition \ref{def:SourceTGW}).
Though there is no principle difficulty to mathematically solving for the effect of this jump \cite{ZonalWindsWicht2}, in reality such a discontinuity would not persist and seems unphysical.
Another option for continuing the surface winds inwards is to do it radially, meaning $u^{\proj}$ is constant along rays emanating from the center of the planet,
but this approach might also be unphysical.
The problem of the equatorial discontinuity is further discussed in \cite*{ZonalWindsWicht2} and \cite*{OddCoeffsDiscontKong}.
\par
Some authors suggest that the winds may penetrate deeper at lower latitude \cite{CylindWindsKaspi}. 
On the other hand, \cite{ZonalWindsWicht2} assume that equatorially asymmetric winds at low latitudes remain shallow, which also avoids the problem of the discontinuity.
\par
Furthermore, at least for Saturn it seems advisable to relax the condition of a spherically symmetric solution, since this gas giant deviates quite far from a spherical shape \cite{SaturnFactSheet}. 
This would require a completely new approach and it is unclear whether any ansatz functions could be found that would make a non-spherical problem assessible to (semi) analytical methods. 
\par 
Regardless of these problems or possible generalizations, this article provides a mathematically rigorous way to study the atmospheric structure of gas giants using gravitational data. 
The proofs for the existence and properties of the adapted orthonormal basis in Section \ref{sec:ONBGravPot}, as well as the calculation method from Section \ref{sec:DetDynPot} are to the best knowledge of the authors, new developments.

\appendix

\section{Further Properties of Bessel functions}\label{sec:BesselProp}
This section summarizes some theorems and methods from \cite*{AaS,BesselWatson} that are 
needed to prove those properties of the functions $B_{m,n}$ in the previous 
chapters.
\begin{definitionS}\label{def:BesselFunc}
    The \textbf{Bessel functions} of the first kind and of order $\nu\geq -1/2$ are 
    defined via 
    \begin{equation*}
        J_{\nu}(x):=\sum_{m=0}^{\infty}\frac{(-1)^m}{m!\, \Gamma(\nu+m+1)}
        \left(\frac{x}{2}\right)^{2m+\nu}
    \end{equation*}
    for every $x\in\R$. The \textbf{spherical Bessel functions} of the first kind and of 
    order $n\in\N_0\cup\{-1\}$ are then defined by
    \begin{equation*}
        j_n(x) := \sqrt{\frac{\pi}{2x}}J_{n+1/2}(x)
    \end{equation*}
    for every $x>0$.
\end{definitionS}
\begin{lemmaS}\label{lma:RecurrenceBessel}
    For every $\nu\geq\frac{1}{2}$ and $x\in\R$ we have
    \begin{equation}\label{eq:RecurrenceBessel}
        J_{\nu}'(x) = J_{\nu-1}(x)-\frac{\nu}{x}J_{\nu}(x).
    \end{equation}
\end{lemmaS}
\begin{proof}[Proof]
    The proof is omitted, but this is stated as Equation 9.1.27 on page 361 in \cite*{AaS}.
\end{proof}
\begin{theoremS}\label{thm:NormalizedBessel}
    If $n\in\N_0$ is fixed but arbitrary and the sequence $(\lambda_{n,m})_{m\in\N}$ is 
    made up of all the positive solutions of 
    \begin{equation*}
        xJ_{n+1/2}'(x)+\left(n+\frac{1}{2}\right)J_{n+1/2}(x) = 0,
    \end{equation*}
    then we can calculate the following integral as
    \begin{equation*}
        \int_{0}^{1}sJ_{n+1/2}^2(\lambda_{n,m}s)\di s = 
        \frac{J_{n+1/2}^2(\lambda_{n,m})}{2}.
    \end{equation*}
\end{theoremS}
\begin{proof}[Proof]
    We apply Equation 11.4.5 on page 485 in \cite*{AaS}. In the notation of \cite*{AaS} we have 
    $\nu=n+\frac{1}{2}$, $a=n+\frac{1}{2}$ and $b=1$. Thus we get the desired relation.
\end{proof}
\begin{theoremS}\label{thm:SeriesContBessel}
    If $f$ is continuously differentiable on $[0,1]$, $n\in\N_0$ is fixed but arbitrary 
    and the sequence $(\lambda_{n,m})_{m\in\N}$ is made up of every positive solution of 
    the equation 
    \begin{equation}\label{eq:FreqEqBessel}
        xJ_{n+1/2}'(x)+\left(n+\frac{1}{2}\right)J_{n+1/2}(x) = 0,
    \end{equation}
    then we can write $f$ as the Dini-Bessel series
    \begin{equation}\label{eq:BesselSeries}
        f(s) = \sum_{m=1}^{\infty}c_m\frac{2}{J_{n+1/2}^2(\lambda_{n,m})}
        J_{n+1/2}(\lambda_{n,m}s).
    \end{equation}
    The convergence of the series is uniform in $[\varepsilon,1-\varepsilon]$ for every 
    $0<\varepsilon<1$ and coefficients $c_m$ are given by 
    \begin{equation}\label{eq:DefCoeffs}
        c_m := \int_{0}^{1}sf(s)J_{n+1/2}(\lambda_{n,m}s)\di s.
    \end{equation}
\end{theoremS}
\begin{proof}[Proof]
    Since $f\in\Cont^1[0,1]$, it is of bounded variation and continuous on $[0,1]$.
    This together with the fact that the $\lambda_{n,m}$ satisfy Equation 
    \eqref{eq:FreqEqBessel} means we can apply Theorem 18.33 on pages 600-602 from 
    \cite*{BesselWatson} (with $\nu=n+1/2$) and conclude that 
    \begin{align*}
        f(s) = \mathcal{B}_0(s)+\sum_{m=1}^{\infty}b_mJ_{n+1/2}(\lambda_{n,m}s),
    \end{align*}
    where the convergence is uniform on $[\varepsilon,1-\varepsilon]$ for every 
    $0<\varepsilon<1$. The notation is taken from \cite*{BesselWatson}. In our case the 
    function $\mathcal{B}_0$ is identically $0$, since the sum of the coefficients in 
    front of $J_{n+1/2}$ and the order of the involved Bessel functions in 
    Equation \eqref{eq:FreqEqBessel} is $2n+1>0$ 
    (see page 597 in 18.3 of \cite*{BesselWatson} for a definition of $\mathcal{B}_0$). 
    The coefficients $b_m$ are given by the equation 
    \begin{align*}
        b_m\int_{0}^{1}sJ_{n+1/2}^2(\lambda_{n,m}s)\di s = \int_{0}^{1}sf(s)J_{n+1/2}
        (\lambda_{n,m}s)\di s,
    \end{align*}
    which is taken from page 597 in \cite*{BesselWatson}.
    The integral on the left-hand side of the above equation can be calculated with 
    Theorem \ref*{thm:NormalizedBessel}, yielding
    \begin{equation*}
        b_m = \frac{2}{J_{n+1/2}^2(\lambda_{n,m})}c_m.\qedhere
    \end{equation*}
\end{proof}
\begin{lemmaS}\label{lma:AsymptoticCoeffBessel}
    If $n\in\N_0$ and $f$ is a function such that $s\mapsto\sqrt{s}f(s)$ is continuously 
    differentiable on $[0,1]$, we have 
    \begin{equation*}
        \int_{0}^{1}sf(s)J_{n+1/2}(\lambda s)\di s = \mathcal{O}
        \left(\lambda^{-3/2}\right)\quad\text{as }\lambda\to\infty.
    \end{equation*}
\end{lemmaS}
\begin{proof}[Proof]
    If $s\mapsto \sqrt{s}f(s)$ is continuously differentiable on $[0,1]$, it is of 
    bounded variation on $[0,1]$. This means we can apply the lemma proven in 18.27 on 
    page 595 in \cite*{BesselWatson}.
\end{proof}
\begin{lemmaS}\label{lma:AsymptoticBesselFunc}
    For every $n\in\N_0$ the asymptotic expansion 
    \begin{equation*}
        J_{n+1/2}(x) = \left(\frac{2}{\pi x}\right)^{1/2}\left(\cos\left(x-
        \pi\frac{n+1}{2}\right)+\mathcal{O}\left(\frac{1}{x}\right)\right)
    \end{equation*}
    holds as $x\to\infty$.
\end{lemmaS}
\begin{proof}[Proof]
    This is Theorem 7.21 from page 199 in \cite*{BesselWatson} applied for $\nu=n+1/2$.
\end{proof}
\begin{lemmaS}\label{lma:AsymptoticBesselZero}
    Let $n\in\N_0$ be fixed but arbitrary and let $(\lambda_{n,m})_{m\in\N}$ be the positive 
    zeros of $J_{n-1/2}$, i.e. the solutions of Equation \eqref{eq:FreqEqBessel}. Then the asymptotic 
    expansion
    \begin{equation*}
        \lambda_{n,m} = \left(m+\frac{n-1}{2}\right)\pi+
        \mathcal{O}\left(\frac{1}{m}\right)
    \end{equation*}
    holds as $m\to\infty$.
\end{lemmaS}
\begin{proof}[Proof]
    We use Theorem 15.53 on pages 505-507 in \cite*{BesselWatson} with $\nu=n-1/2$. This fact is 
    also stated in Equation 9.5.12 on page 371 in \cite*{AaS}.
\end{proof}
With the three preceeding lemmata we can now prove a result about the decay of the 
Fourier coefficients from Chapter \ref*{sec:ONBGravPot}.
\begin{theoremS}\label{thm:DecayFourierCoeffBessel}
    For a fixed but arbitrary $n\in\N_0$ and a function $f\in\Cont^1[0,1]$, let  
    $(\lambda_{n,m})_{m\in\N}$ be the positive solutions to Equation 
    \eqref{eq:FreqEqBessel}. Then there exists a positive integer $M\in\N$ and a 
    constant $C>0$ (both of which may depend on $f$ and $n$), such that for every 
    $m\geq M$ we have
    \begin{equation*}
        \left|\sqrt{\frac{2}{j_n^2(\lambda_{n,m})}}\int_{0}^{1}s^2f(s)j_n(\lambda_{n,m}s)
        \di s\right|\leq \frac{C}{m-3/4}.
    \end{equation*}
\end{theoremS}
\begin{proof}[Proof]
    We define the integral
    \begin{equation*}
        d_m := \int_{0}^{1}s^2f(s)j_n(\lambda_{n,m}s)\di s,
    \end{equation*}
    from which we can derive 
    \begin{equation*}
        \left|\sqrt{\frac{2}{j_n^2(\lambda_{n,m})}}d_m\right| = 
        \frac{\sqrt{2}}{|J_{n+1/2}(\lambda_{n,m})|}\left|\int_{0}^{1}s
        \left(s^{1/2}f(s)\right)J_{n+1/2}(\lambda_{n,m}s)\di s\right|.
    \end{equation*}
    By Lemma \ref*{lma:AsymptoticBesselZero} we know that $\lambda_{n,m}\to\infty$ for 
    $m\to\infty$. We want to apply Lemma \ref*{lma:AsymptoticCoeffBessel} to the function 
    $g(s):=\sqrt{s}f(s)$. For this $s\mapsto\sqrt{s}g(s)=sf(s)$ has to be continuously 
    differentiable on $[0,1]$, which is obviously true. Now we apply all of the 
    asymptotic formulas which we derived previously, at first
    \begin{equation*}
        \left|\int_{0}^{1}s\left(s^{1/2}f(s)\right)J_{n+1/2}(\lambda_{n,m}s)\di s\right|
        \leq C_1\lambda_{n,m}^{-3/2}
    \end{equation*}
    holds for some constant $C_1$ and large values of $m\in\N$. Secondly for large values 
    of $m$, the coefficient $\lambda_{n,m}$ lies between
    $m\pi +((n-1)/2)\pi-\pi/4$ and $m\pi +((n-1)/2)\pi+\pi/4$ by Lemma 
    \ref*{lma:AsymptoticBesselZero}, which means 
    $|\cos(\lambda_{n,m}-\pi(n+1)/2)|\geq 1/\sqrt{2}$ holds. Thirdly, we can use Lemma 
    \ref*{lma:AsymptoticBesselFunc}, the reverse triangle inequality and the previous
    consideration about the $\lambda_{n,m}$ to derive the lower bound
    \begin{equation*}
        |J_{n+1/2}(\lambda_{n,m})|\geq \left|\frac{2}{\pi\lambda_{n,m}}\right|^{1/2}
        \left|\left|\cos\left(\lambda_{n,m}-\pi\frac{n+1}{2}\right)\right|-
        \left|\mathcal{O}\left(\frac{1}{\lambda_{n,m}}\right)\right|\right|.
    \end{equation*}
    For large $m$ the term $\mathcal{O}(\lambda_{n,m}^{-1})$ will be smaller than 
    $1/\sqrt{8}$. This implies that
    \begin{equation*}
        |J_{n+1/2}(\lambda_{n,m})|\geq\frac{1}{2}
        \left(\frac{1}{\pi\lambda_{n,m}}\right)^{1/2}.
    \end{equation*}
    Putting this all together we conclude that for large values of $m\in\N$
    \begin{equation*}
        \left|\sqrt{\frac{2}{j_n^2(\lambda_{n,m})}}d_m\right|\leq
        \frac{C_2}{\lambda_{n,m}}
    \end{equation*}
    holds for some constant $C_2$. We also know that 
    $\lambda_{n,m}\geq m\pi+\pi(n-1)/2-\pi/4\geq \pi(m-3/4)$, which proves that there is 
    some constant $C>0$ such that 
    \begin{equation*}
        \left|\sqrt{\frac{2}{j_n^2(\lambda_{n,m})}}d_m\right| 
        \leq \frac{C}{m-3/4}
    \end{equation*}
    holds for large values of $m$.
\end{proof}

\emergencystretch=1em
\printbibliography[heading=bibintoc,title=References]
\newpage
\end{document}